\documentclass[conference]{IEEEtran}\IEEEoverridecommandlockouts
\usepackage{etoolbox}
\usepackage{algorithmic}
\usepackage{algorithm}

\makeatletter
\patchcmd{\@makecaption}
  {\scshape}
  {}
  {}
  {}
\makeatletter
\patchcmd{\@makecaption}
  {\\}
  {.\ }
  {}
  {}
\makeatother
 \usepackage{amsmath,amssymb}
 \usepackage{subfigure}
 \usepackage{graphicx,graphics,color,psfrag}
 \usepackage{cite,balance}
 \usepackage{caption}
 \allowdisplaybreaks
 \usepackage{algorithm}
 \usepackage{accents}
 \usepackage{amsthm}
 \usepackage{bm}
 \usepackage{algorithmic}
 \usepackage[english]{babel}
 \usepackage{multirow}
 \usepackage{enumerate}
 \usepackage{graphicx}
 \usepackage{subfigure}
 \usepackage[T1]{fontenc}
 \usepackage{aecompl}

 \usepackage{cases}
 \usepackage{bm}
 \usepackage{stfloats}
 \usepackage{dsfont}
 \usepackage{color,soul}
 \usepackage{amsfonts}
 \usepackage{cite,graphicx,amsmath,amssymb}
 \usepackage{subfigure}
 \usepackage{fancyhdr}
 \usepackage{hhline}
 \usepackage{graphicx,graphics}
 \usepackage{array,color}
 \usepackage{amsmath}
\usepackage{float}
\usepackage{amssymb}
\usepackage{amsmath}
\usepackage{amsthm}
\usepackage{amsfonts}
\usepackage{graphicx}
\usepackage{utfsym}
\usepackage{epstopdf}
\usepackage{cite}
\usepackage{amsmath,bm}
\usepackage{subfigure}
\usepackage{graphicx}
\usepackage{color}
\usepackage{graphicx}
\usepackage{amssymb}
\usepackage{calc}
\usepackage{caption}
\newtheorem{definition}{Definition}

\newtheorem{proposition}{\textbf{Proposition}}
\newtheorem{remark}{\textbf{Remark}}

\begin{document}

\title{MambaJSCC: Adaptive Deep Joint Source-Channel Coding with Generalized State Space Model}% \thanks{1}}
\author{Tong Wu, Zhiyong Chen, Meixia Tao, \emph{Fellow, IEEE}, Yaping Sun, Xiaodong Xu, \\Wenjun Zhang, \emph{Fellow, IEEE}, and Ping Zhang, \emph{Fellow, IEEE}
% 		$^{*}$Cooperative Medianet Innovation Center, Shanghai Jiao Tong University, Shanghai, China\\
% $\dag$Beijing University of Posts and Telecommunications, Beijing, China\\   
% 		Email: \{wu\textunderscore tong, zhiyongchen, mxtao, zhangwenjun\}@sjtu.edu.cn, xuxd@pcl.ac.cn, pzhang@bupt.edu.cn}
\thanks{The paper will be presented in part at IEEE GLOBECOM 2024 \cite{gc-MambaJSCC}.}
\thanks{T. Wu, Z. Chen, M. Tao and W. Zhang are with the Cooperative Medianet Innovation Center (CMIC), Shanghai Jiao Tong University, Shanghai 200240, China, and Shanghai Key Laboratory of Digital Media Processing and Transmission (e-mail: \{wu\_tong, zhiyongchen, mxtao, zhangwenjun\}@sjtu.edu.cn). (\textit{Corresponding author: Zhiyong Chen})}
\thanks{Y. Sun is with the Department of Broadband Communication, Pengcheng Laboratory, Shenzhen 518055, China (e-mail: sunyp@pcl.ac.cn).}
\thanks{X. Xu and P. Zhang are with the State Key Laboratory of Networking and Switching Technology, Beijing University of Posts and Telecommunications, Beijing 100876, China, and also with the Department of Broadband Communication, Peng Cheng Laboratory, Shenzhen 518055, China (e-mail: xuxd@pcl.ac.cn; pzhang@bupt.edu.cn).}}
	
\maketitle
\begin{abstract}
Lightweight and efficient neural network models for deep joint source-channel coding (JSCC) are crucial for semantic communications. In this paper, we propose a novel JSCC architecture, named MambaJSCC, that achieves state-of-the-art performance with low computational and parameter overhead. MambaJSCC utilizes the visual state space model with channel adaptation (VSSM-CA) blocks as its backbone for transmitting images over wireless channels, where the VSSM-CA primarily consists of the generalized state space models (GSSM) and the zero-parameter, zero-computational channel adaptation method (CSI-ReST). We design the GSSM module, leveraging reversible matrix transformations to express generalized scan expanding operations, and theoretically prove that two GSSM modules can effectively capture global information. We discover that GSSM inherently possesses the ability to adapt to channels, a form of endogenous intelligence. Based on this, we design the CSI-ReST method, which injects channel state information (CSI) into the initial state of GSSM to utilize its native response, and into the residual state to mitigate CSI forgetting, enabling effective channel adaptation without introducing additional computational and parameter overhead. Experimental results show that MambaJSCC not only outperforms existing JSCC methods (e.g., SwinJSCC) across various scenarios but also significantly reduces parameter size, computational overhead, and inference delay. 
%In particular, with employing an equal number of the VSSM-CA blocks and the Swin Transformer blocks in SwinJSCC, MambaJSCC achieves a $0.58$ dB gain in peak-signal-to-noise ratio (PSNR) while requiring only 72\% multiply-accumulate operations, 51\% of the parameters, and 91\% of ID.
\end{abstract}


%\begin{IEEEkeywords}
%Joint Source-Channel Coding, Semantic Communications, State Space Model, Swin Transformer
%\end{IEEEkeywords}

\section{Introduction}
%\vspace{-0.011 cm}
Deep learning-based joint source-channel coding (JSCC), which integrates source and channel coding into a single process using neural networks, has recently garnered significant attention as a method to enhance the efficiency of information transmission over wireless channels. Leveraging deep learning, JSCC can effectively extract and utilize the semantic features of data, which are critical for understanding the meaning of information in task-specific contexts. Consequently, JSCC is emerging as a key technology for enabling semantic communications \cite{Qin2,Qin1}, a pivotal element in the development of sixth-generation wireless networks\cite{CDDM,Bo,Dai}.

Recent research on JSCC has primarily focused on finding suitable approaches to maximize end-to-end performance. For text transmission, \cite{text} proposes a JSCC framework with a joint iterative decoder to improve quality. For video transmission, \cite{video} introduces a deep video semantic transmission system based on JSCC to optimize overall performance. For specialized content types, such as point cloud, \cite{pointcloud} proposes an appropriate JSCC architecture. Channel state information (CSI) estimated at the receiver is transmitted back to the transmitter utilizing JSCC in \cite{JSCC-CSI}. For image transmission, DeepJSCC, proposed in \cite{gundu2019}, relies on convolutional neural networks (CNNs) and outperforms traditional separation-based schemes, e.g., JPEG2000 and capacity-achieving channel code. To further improve JSCC performance, \cite{ViT-MIMO, NTSCC} develop the backbone module of JSCC from CNN to Vision Transformer (ViT). Furthermore, with the emergence of the Swin Transformer\cite{Ze},  SwinJSCC has been proposed in \cite{SwinJSCC}, which replaces CNNs with Swin Transformers. It is adopted in \cite{NTSCC+} to replace ViT as an improved version of \cite{NTSCC}.

Meanwhile, channel adaptation methods, which help mitigate performance degradation caused by the mismatch of channel states between training and evaluation stages, are essential for improving the performance of a single model across a wide range of channel conditions. These methods also reduce the training burden and save parameters of multiply JSCC models. ADJSCC, proposed in \cite{ADJSCC}, utilizes the attention feature modules to incorporate the signal-to-noise (SNR) into the coding process. Hyper-AJSCC is introduced in \cite{Hyper-AJSCC}, which employs a hypernetwork to adjust the parameters based on the SNR, achieving similar performance to ADJSCC with significantly fewer parameters and computational overhead. A Gated net is proposed in \cite{gate-adaptive} that filters semantics according to varying SNRs. These works focus on channel adaptation methods for CNN-based models. For Transformer-based models, \cite{SwinJSCC} introduces a plug-in module called Channel ModNet. A semantic-aware adaptive channel encoder is developed in \cite{GRACE} for integrating SNR. Additionally, \cite{CSI-embedding} proposes the CSI Embedding method, integrating multi-user SNRs into the Swin Transformer blocks in each coding stage, achieving similar performance to Channel ModNet with far fewer parameters and computational overhead.%, though it is limited to integer SNRs due to its sine position encoding design.
% \begin{figure}[t]
%   \begin{center}
%     \includegraphics[width=0.45\textwidth]{./figs/model_size.pdf}
%   \end{center}
%     \caption{{The performance versus computational overhead of different JSCC architectures under the AWGN channel on DIV2K dataset with $256$ resolution. The SNR is $10$ dB and the CBR is $\frac{1}{48}$.}}
%     \label{model_size}
%     \vspace{-0.3 cm}
% \end{figure}

Undoubtedly, designing deep learning-based JSCC models involves more than just optimizing end-to-end performance across varying channel states; it also requires balancing the complexity of artificial intelligence (AI) algorithms. This complexity typically includes computational overhead and the number of parameters. Computational overhead directly impacts inference delay (ID), and the larger the number of parameters, the greater the challenge of deploying these models on resource-constrained devices such as edge devices. Performance and complexity are mainly influenced by two core factors: the \textit{backbone model structure} (e.g., CNN or Transformer) and the \textit{channel adaptation method}. For backbone model structures, CNN-based models such as DeepJSCC and ADJSCC are often considered inadequate due to the inherent limitations of CNNs. In contrast, Transformer-based models like SwinJSCC offer superior performance, leveraging the advanced Swin Transformer, though this comes at the cost of higher computational and parameter overhead. Regarding channel adaptation methods, previous approaches have introduced additional modules to integrate CSI, similar to \textbf{plug-in intelligence}, which increases the complexity of a single model. For example, Channel ModNet in SwinJSCC adds an extra $9.87M$ parameters and $3.47G$ multiply-accumulate operations (MACs) for images of $512\times512$ size.
% To this end, we believe that an ideal JSCC scheme should achieve good performance under varying channel conditions, while its AI model maintains low computational complexity and parameter overhead.

%To this end, it is desired to design a novel JSCC scheme for better performance under varying channel states, while maintaining low computational and parameters overhead on the premise of better performance. The novel JSCC is required to maintain its backbone structure to be data-dependent dynamic weight and able to capture global information, which is the core of the great performance of Transformer\cite{Attention}, and also maintain its effective channel adaptation method. Meanwhile, the novel JSCC should encounter the drawbacks of transformers with great computational complexity and massive parameters and the drawbacks of the existing channel adaptation methods, which require additional modules introducing computational and parameter overhead.

More recently, a novel AI model structure called \textbf{Mamba}, proposed in \cite{S6}, has attracted significant attention in the AI community. Mamba is based on selective structure state space models (SSM), which enable it to process sequences with data-dependent dynamic weights in \textbf{linear complexity}. As a result, Mamba achieves remarkable performance in natural language processing tasks while maintaining low computational complexity and parameter overhead compared to Transformers. Mamba-2 \cite{Mamba2} further demonstrates the equivalence between Mamba and linear attention. In the field of computer vision\cite{Survey}, VMamba, proposed in \cite{VSSM}, is used for visual representation learning by cross expanding the image in four directions, achieving promising results in vision tasks. Subsequently, \cite{Zigma, localmamba, EfficientVMamba, PlainMamba} explored various scan directions for improved learning. However, the lack of a unified mathematical expression for all scan expansion techniques limits the mathematical derivation necessary to ensure Vision Mamba's ability to capture global information. Furthermore, to the best of our knowledge, no research has yet explored the integration of Mamba with channel adaptation methods.
% These two challenges present significant obstacles to the design of a novel JSCC model.
%% 接下来说我们的贡献了

To address these challenges, we propose a novel, lightweight, and efficient JSCC method based on Mamba, named \textbf{MambaJSCC}, for image transmission over wireless channels. Specifically, we apply a reversible matrix transformation to unify all the scan expanding operations into a mathematical expression, enabling the SSM to be equipped with arbitrary scan schemes as a generalized SSM (GSSM). Based on this mathematical expression, we theoretically prove that two GSSM modules are sufficient to capture global information, thus achieving superior performance. More interestingly, by analyzing Mamba's unique response to the initial state, \emph{we discover that GSSM inherently possesses the ability to adapt to channels, a form of \textbf{endogenous intelligence}}. Building on this insight, we develop a zero-parameter, zero-computation channel adaptation method, called CSI-ReST, for MambaJSCC. The proposed CSI-ReST method injects the CSI into the initial state of the GSSM. As the state is continuously updated, the CSI gradually fades. To address this, inspired by the ResNet \cite{ResNet} approach, CSI-ReST also injects the CSI into the residual state of the GSSM benefitting from its high-dimension characteristic, enabling channel adaptation without incurring additional computational and parameter overhead. Consequently, the two GSSM modules with CSI-ReST, combined with a door-control structure, form the VSSM-CA modules. Utilizing the powerful VSSM-CA modules as the backbone, MambaJSCC is designed in a hierarchical structure with multi-stages to handle varying numbers of patches of different sizes, aiming for effective encoding and decoding.

The main contributions of the paper are summarized as follows.
\begin{itemize}
  \item \textbf{Architecture Design:} We propose MambaJSCC, a novel JSCC architecture that offers state-of-the-art performance with significantly low computational and parameter overhead. In the MambaJSCC architecture, the VSSM-CA module with GSSM serves as the backbone model structure, and the endogenous intelligence-based CSI-ReST is used as the channel adaptation method. The architecture is hierarchical and multi-stages to handle patches at varying numbers and different resolutions.
  %\item \textbf{Theoretical Analyze:}  We derive the unified mathematical expression for all the scan expanding operations in the reversible matrix transformations form and proposed our GSSM, which is an SSM module equipped with arbitrary generalized scan expanding operations. The unified expression also guides the design if the two GSSM modules through mathematically proof of their ability to capture global information.
  \item \textbf{Theoretical Analysis:} We design the GSSM module, utilizing reversible matrix transformations to describe arbitrary generalized scan expanding operations. Furthermore, we derive the closed-form expression for the output of the GSSM module. For the first time, we theoretically prove that the two GSSM modules with bidirectional scanning, similar to Transformers, are capable of effectively capturing global information.
  \item \textbf{CSI-ReST Method:} We develop the zero-parameter, zero-computation CSI-ReST method to leverage the endogenous capability of GSSM for channel adaptation. This method injects the CSI into the initial state to harness the native response and into the residual state to mitigate CSI forgetting, enabling effective channel adaptation without introducing additional computational and parameter overhead.
  \item \textbf{Experiments Evaluation:} We conduct extensive experiments using diverse data with varying contents and resolutions. The results show that MambaJSCC outperforms all major JSCC architectures in terms of distortion and perception. Furthermore, compared to SwinJSCC with same block number, MambaJSCC exhibits a significantly reduced parameter size, computational overhead, and inference delay e.g., $0.52$ dB gain in peak-signal-to-noise ratio (PSNR) while requiring only 72\% MACs, 51\% of the parameters, and 91\% of ID.

\end{itemize}

The rest of the paper is organized as follows. The preliminary of SSM and the system framework of MambaJSCC are presented in Section \ref{II}. The details of the proposed GSSM mathematical expressions and the CSI-ReST method, along with the structure of VSSM-CA are illustrated in Section \ref{III}. Finally, extensive experimental results are provided in Section \ref{IV}, and conclusions are drawn in Section \ref{V}.

\begin{figure*}[t]
  \begin{center}
    \includegraphics[width=0.95\textwidth]{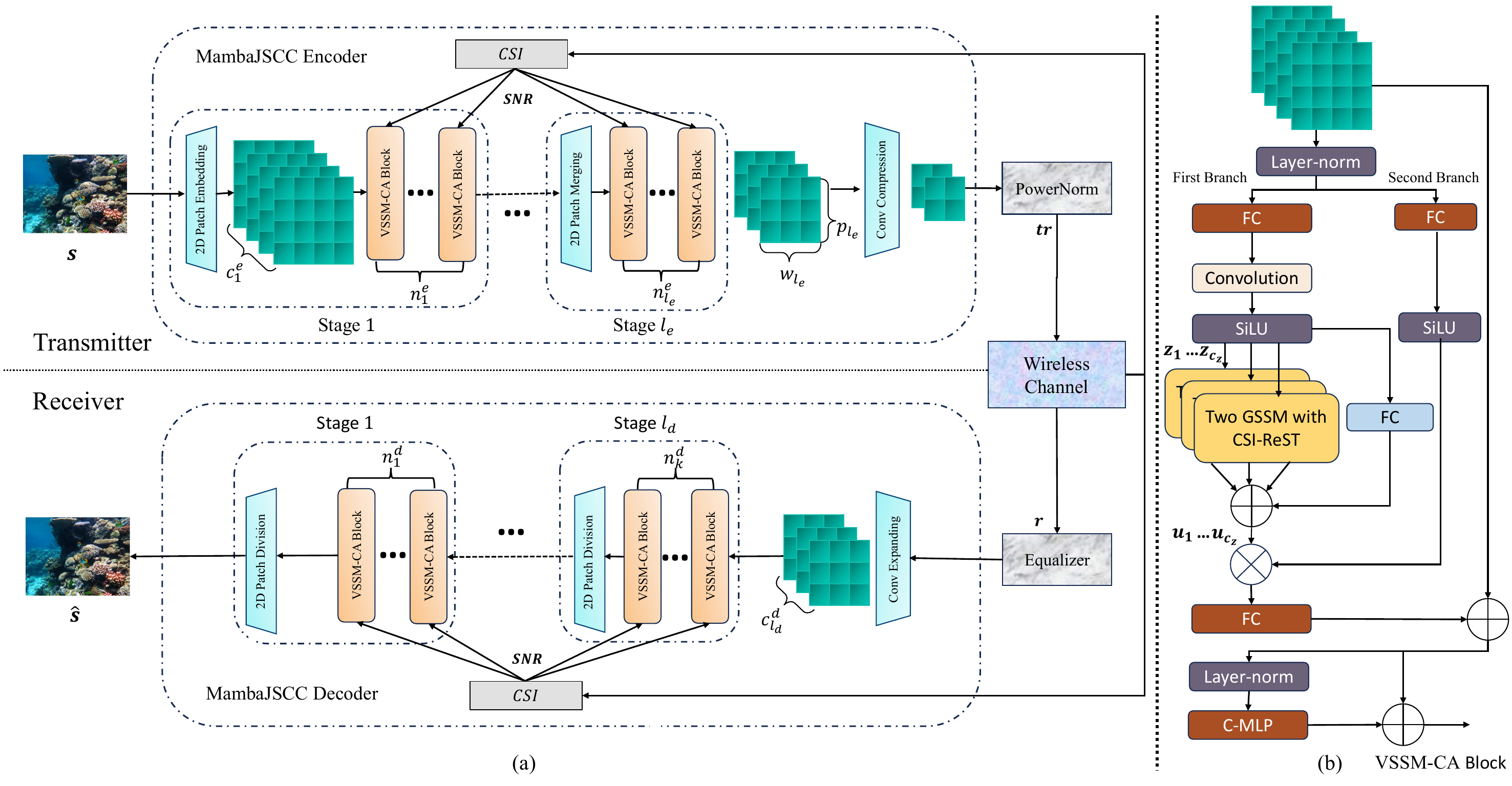}
  \end{center}
    \caption{{(a) The overall architecture of the proposed MambaJSCC. (b) The structure of the VSSM-CA block.}}
    \label{MambaJSCC}
    \vspace{-0.5 cm}
\end{figure*}

% S4 \cite{s4} defines four learnable parameters ($\mathbf{\Delta}, \mathbf{A}, \mathbf{B}, \mathbf{C}$) and acquires the output in an iterative manner through (\ref{D-LTI}), where the $\mathbf{D}$ matrix is set to zero.
%The recent advanced Mamba model draws upon the successful experiences of the Transformer and modifies $\mathbf{\Delta}$, $\mathbf{B}$, and $\mathbf{C}$ into data-dependent learnable parameters, allowing the model to attain a dynamic state space, wherein it can select the most relevant portions from the current state and input as the new state, subsequently utilizing this state to compute outputs. This modification significantly improves the performance of the Mamba.
\section{System Model of MambaJSCC}\label{II}
%In this section, we present a simple preliminary for SSM and then comprehensively illustrate the system framework of the proposed MambaJSCC for wireless image transmission. 
\subsection{Review of State Space Models}
State space models are a class of sequential neural network models inspired by linear time-invariant (LTI) systems and are closely related to CNNs and recurrent neural networks (RNNs). In general, an LTI system maps a one-dimension (1D) function or sequence $x(t)\in\mathbb{R}$ to another one-dimension function or sequence $y(t)\in\mathbb{R}$ through a latent state space $\mathbf{h}(t)\in\mathbb{R}^N$. Here, $N$ represents the dimension of the latent state space. This process can be formulated using linear ordinary differential equations (ODEs) as follows:
\begin{align}\label{LTI}
  {\mathbf{h}^{\prime}}(t)=\tilde{\mathbf{A}}\mathbf{h}(t)+\tilde{\mathbf{B}}x(t),~~ y(t)=\tilde{\mathbf{C}}\mathbf{h}(t),
\end{align} 
where $\tilde{\mathbf{A}}\in \mathbb{R}^{N\times N}$, $\tilde{\mathbf{B}}\in \mathbb{R}^{N\times 1}$, and $\tilde{\mathbf{C}}\in \mathbb{R}^{1\times N}$ are the state matrix, input matrix, and output matrix, respectively.

The OEDs represent continuous-time systems, allowing the matrices to naturally become time-varying. Meanwhile, to integrate these continuous-time systems into deep learning frameworks, a discretization process is required. A widely-adopted method for discretization is the zero-order hold (ZOH) rule. The time-varying and discretized version of $\tilde{\mathbf{A}}$, $\tilde{\mathbf{B}}$, $\tilde{\mathbf{C}}$ are expressed as $\mathbf{A}=\mathbf{A}_{1:T} \in \mathbb{R}^{T\times N\times N}$, $\mathbf{B}=\mathbf{B}_{1:T}\in \mathbb{R}^{T\times N\times 1}$, $\mathbf{C}=\mathbf{C}_{1:T}\in \mathbb{R}^{T\times 1\times N}$. The sequence $x(t)$ is discretized as $\mathbf{x} \in \mathbb{R}^T$. The corresponding system output is $\mathbf{y} \in \mathbb{R}^T$, and the discretized hidden state is represented by $\mathbf{h}\in \mathbb{R}^{T\times N}$, where $T$ is the length of the sequence.
In this paper, we use the subscript $t\in{1,2,...,T}$ on these matrices and sequences to represent their values at time $t$. Therefore, (\ref{LTI}) can be rewritten in its discrete, and time-varying form as follows \cite{Mamba2}:
% \begin{align}
%   \mathbf{\bar{A}}=\exp(\mathbf{\Delta}\mathbf{A}),~~
%  \mathbf{\bar{B}}=(\mathbf{\Delta}\mathbf{A})^{-1}(\exp(\Delta\mathbf{A})-\mathbf{I})\mathbf{\Delta}\mathbf{B}.
% \end{align}
\begin{align}\label{SSM}
  \mathbf{h}_{t}=\mathbf{{A}}_t\mathbf{h}_{t-1}+\mathbf{B}_tx_t,~~
 y_t=\mathbf{C}_t\mathbf{h}_t.
\end{align}

For two-dimension (2D) sources, such as images, a commonly used and well-developed method for SSM is the scan expanding and scan recovery method\cite{VSSM,Survey,localmamba,PlainMamba,EfficientVMamba,Zigma}. This method flattens a 2D patch to several 1D sequences in designated directions, followed by SSMs. After that, these sequences are reconstructed back into 2D patches in the corresponding directions, and merged by adding together into one patch.

\subsection{An Overview of MambaJSCC}
The overall architecture of MambaJSCC is illustrated in Fig. \ref{MambaJSCC}(a). In the transmitter, the input image $\mathbf{s}\in\mathbb{R}^{3 \times P \times W}$ is encoded into the channel input signal $\mathbf{q}\in\mathbb{C}^{c \times p \times w}$ by the MambaJSCC encoder, where $P$ and $W$ represent the height and width of the source image and $c$, $p$, $w$ denote the patch number, height and weight of the channel input signal, respectively. This encoding process involves $l_e$ coding stages and a convolution compression module.

In the first coding stage, the input image $\mathbf{s}$ is initially fed into a 2D patch embedding module, which outputs $c_1^e$ patches, each of size $p_1\times w_1$. These patches are processed by $n_1^e$ VSSM-CA blocks, which incorporate the CSI to adaptively extract and encode features from the patches, while maintaining their number and size. The output patches are then fed into Stage 2, where a patch merging module is applied. This module merges the $c_1^e$ patches of size $p_1 \times w_1$ into a greater number of patches $c_2^e$, with a smaller size $p_2 \times w_2$, to capture more refined information. Following the merging, $n_2^e$ VSSM-CA blocks with CSI are employed to process these patches. This process is repeated for ($l_e-1$) stages in the MambaJSCC encoder, where each stage $k$ $(k=2,3,...,l_e)$ involves the patch merging module producing $c_k^e$ patches of size $p_k \times w_k$, followed by $n_k^e$ VSSM-CA blocks processing the patches. At the final stage of the encoder, the output patches are fed into a convolution compression module, where a CNN compresses the patches to the designated number of channel uses. After compression, the patches are normalized for power and transmitted over the wireless channel, which is assumed to be either an additive white Gaussian noise (AWGN) channel or a Rayleigh fading channel. 

At the receiver, the MambaJSCC decoder first applies an equalizer, such as the minimum mean square error (MMSE) equalizer, to equalize the received signal $\mathbf{r}$ in the case of the Rayleigh fading channel. It then decodes the signal into $\mathbf{\hat{s}} \in \mathbb{R}^{3 \times P \times W}$ for reconstruction. The MambaJSCC decoder consists of $l_d$ decoding stages and a convolution expanding module. Specifically, after equalization, the convolution expanding module expands the equalized signal into $c^d_{l_d}$ 2D patches, each with size $p_{l_d} \times w_{l_d}$, matching the output patches from Stage $l_e$ in the encoder. These patches are then fed into $l_d$ concatenated stages, denoted as Stage $k,k=l_d,l_d-1,...,1$. For $k \ge 2$, each stage initially decodes the patches using $n_k^d$ VSSM-CA blocks along with the CSI, which is obtained via channel estimation. The output patches are then divided into fewer but larger patches $c_{k-1}^d$ with size $p_{k-1} \times w_{k-1}$. This step is performed by a 2D patch division module that first applies layer normalization to the input patches, followed by a full-connection (FC) layer that adjusts the number of patches from $c^d_k$ to $c^d_{k-1}\times\frac{p_{k-1}w_{k-1}}{p_kw_k}$. The goal of expanding the number of patches is to organize every $\frac{p_{k-1}w_{k-1}}{p_kw_k}$ patches into $c_{k-1}^d$ groups. Within each group, a pixel shuffle operation is performed to cyclically select elements, creating a new patch that is with $\frac{p_{k-1}w_{k-1}}{p_kw_k}$ times larger size. As a result, the module outputs $c^d_{k-1}$ new patches, each with a larger size of $p_{k-1}\times w_{k-1}$. In Stage 1, the 2D patch division module outputs three patches, each with size $P \times W$ for reconstruction.

The MambaJSCC encoder and decoder are jointly trained with the objective of minimizing the distortion between the input image and the reconstructed image. Therefore, the loss function $L$ is formulated as follows:
\begin{align}
  L=d(\mathbf{s},\mathbf{\hat{}s}),
\end{align}
where $d(\cdot)$ is the distortion or perception loss function.

\section{The Design of VSSM-CA module}\label{III}
In this section, we discuss the two core technologies for the VSSM-CA module: the GSSM and the CSI-ReST channel adaptation method. Using these two core technologies, we present the overall structure of the VSSM-CA module.
%In this section, we provide the structure of the proposed VSSM-CA block, which is the core component of our MambaJSCC. We detailed illustrate the two main module of the VSSM-CA block: the VSSM module and the proposed CSI embedding module for channel adaptive coding.  
\subsection{The GSSM Module} %可能这里说一句关于复杂度就可以了，因为是利用Mamba的低复杂度
To develop the GSSM module for capturing global information, we establish a unified mathematical expression for the SSM module with arbitrary extended scan expanding and recovery operations based on matrix transformation.
\begin{definition}
   \textbf{SSM operator} is defined as a sequence transformation, which is a parameterized mapping from a sequence $\mathbf{x} \in \mathbb{R}^T$ to another sequence $\mathbf{y} \in \mathbb{R}^T$ as defined in (\ref{SSM}). This mapping is denoted as $\mathbf{y}=SSM(\mathbf{A},\mathbf{B},\mathbf{C})(\mathbf{x})=SSM(\mathbf{A}_{1:T},\mathbf{B}_{1:T},\mathbf{C}_{1:T})(\mathbf{x})$, where $\mathbf{A}_{1:T},\mathbf{B}_{1:T}$ and $\mathbf{C}_{1:T}$ are the parameters of the operator, and $T$ refers to the sequence length.
\end{definition}
Thus, the operations in the SSM modules are represented by an operator parameterized by $\mathbf{A}_{1:T},\mathbf{B}_{1:T}$ and $\mathbf{C}_{1:T}$.

\begin{definition}
  \textbf{Matrix transformation} is defined as a special case of sequence transformation, where the transformation can be written in the form $\mathbf{y}=\mathbf{M_\theta}\mathbf{x}$, with $\mathbf{M}_\theta$ being a matrix derived from the parameters $\theta$. For simplicity and clarity, we denote the matrix as $\mathbf{M_\theta}$, omitting the parameters $\theta$ when they are clear from context.
\end{definition}
Accordingly, the relationship between the SSM operator with $\mathbf{h}_0=\mathbf{0}$ and the matrix transformation is established in \cite{Mamba2}, which is given by 
  \begin{equation}
    \mathbf{y}=SSM(\mathbf{A},\mathbf{B},\mathbf{C})(\mathbf{x})=\mathbf{Mx},\label{SSM-0}\\
  \end{equation}
where the element in the $i$-th row and $j$-th column of $\mathbf{M}$ is derived as $\mathbf{C}_i\mathbf{A}^\times_{i:j}\mathbf{B}_j$. Here, $\mathbf{I}$ is the identity matrix, and 
\begin{align}\label{Atimes}
\mathbf{A}^\times_{i:j}=\left\{
\begin{aligned}
&0, \  \ \ i<j,\\
&\mathbf{I}, \ \ \  i=j,\\
&\mathbf{A}_i \mathbf{A}_{i-1} \cdot \cdot \cdot \mathbf{A}_{j+1}, \ \ \  i>j.
\end{aligned}
\right.
\end{align}

The key technology in extending 1D SSM to 2D vision SSM (VSSM) for vision tasks is the scan expansion approach. Using this approach, VSSM first unfolds the 2D visual features into several 1D sequences. These sequences are then processed by the SSM operator and subsequently reassembled into the 2D feature. The scan directions dictate the unfolding process, playing a crucial role in VSSM's ability to effectively learn and capture visual features.

To investigate the influence of different scan directions on the VSSM module, it is necessary to define a unified mathematical expression for both the scan directions and  the SSM. By analyzing their characteristics, we can observe that the scan expanding operation rearranges the elements of a 2D feature into a 1D sequence, with the scan directions determining the positions of the elements in the sequence. Thus, the essence of scan directions is the reordering of vectorized features. This reordering can be viewed as the application of a finite number of row-swapping transformations to the vectorized feature. From this aspect, the scan expanding process in a specific direction can be expressed as a combination of a vectorization operation and a matrix transformation derived from finite row-swapping operations. We thus have the following proposition.

\begin{proposition}
  The \textbf{Scan Expanding} from a 2D feature $\mathbf{Z}$ to a 1D sequence $\mathbf{x^\zeta}$ in the direction $\zeta $ can be formulated as 
  \begin{align}\label{expand}
    \mathbf{x}^\zeta=\mathbf{R}_\zeta vec(\mathbf{Z})=\mathbf{R}_\zeta \mathbf{z},
  \end{align}
where $\mathbf{R}_\zeta$ is derived by transforming the identity matrix with a finite number of row swapping operations and $\mathbf{z}$ is the normally vectorized form of $\mathbf{Z}$.
\end{proposition}
\begin{proof}
Based on the aforementioned scan expansion operation, this proposition can be easily derived.
\end{proof}
The VSSM module processes the sequence $\mathbf{x^\zeta}$ using the SSM operator, resulting in $\mathbf{y}$, which can be formulated as:
%With the experssion, we can combine the scan expanding with the SSM operator in a matrix transformation as
\begin{align}\label{SSMop}
  \mathbf{y}=SSM(\mathbf{A},\mathbf{B},\mathbf{C})(\mathbf{x}^\zeta).
\end{align}

Next, $\mathbf{y}$ needs to be reordered back to its original sequence before being reshaped into the 2D output. This recovery operation can also be expressed as a matrix transformation:
\begin{align}\label{recover}
  \mathbf{y^\zeta}=\mathbf{R}_\zeta^T \mathbf{y}.
\end{align}
Finally, $\mathbf{y}^\kappa$ is reshaped as the 2D feature output $\mathbf{Y}^\kappa$.
\begin{remark}
The inverse of $\mathbf{R}_\zeta$ is used to restore the order of $\mathbf{y}$ to $\mathbf{y^\zeta}$. In (\ref{recover}), we apply $\mathbf{R}^T_\zeta$ for the recover process because the $\mathbf{R}_\zeta$ is obtained by applying a finite number of row-swapping operations to the identity matrix. Therefore, the inverse of $\mathbf{R}_\zeta$ is $\mathbf{R}_\zeta^T$.
\end{remark}
Upon reviewing the computational flow of the VSSM module, we observe that it can be regarded as a sequence mapping module with a normal vectorization module and a module normally reshaping a vector to a matrix. The core component of the VSSM module is the sequence mapping module, which determines its key characteristics. In the VSSM module, the sequence mapping module can be seen as an SSM operator combined with a matrix $\mathbf{R}_\zeta$ to reorder the input and output sequence. Furthermore, we find from (\ref{expand}) and (\ref{recover}) that the only requirement for the sequence mapping module is that the input reordering operation must be reversible. As a result, we define GSSM to provide a more universal framework for understanding and designing the VSSM module.
%For simplification and clarification, we consider the vectorization as an additional module and analyse the characteristics in sequence form.
%In this scene, the way to expand the 2D feature should be called as scan scheme.
%The first step is to generalized the $\mathbf{R}_\zeta$.

We generalize the definition of $\mathbf{R}_\zeta$ by deriving it not only from row-swapping transformations but also from arbitrary elementary transformations, resulting in an invertible matrix.
\begin{definition}
  The \textbf{scan exchanging} from a sequence $\mathbf{z}$ to a sequence $\mathbf{x^ \kappa }$ in the scheme $ \kappa  $ can be formulated as follows:
  \begin{align}\label{gexpand}
    \mathbf{x^\kappa}=\mathbf{R}_\kappa \mathbf{z},
  \end{align}
  where $\mathbf{R}_\kappa$ is an invertible matrix that depends on the scan scheme $\kappa$.
\end{definition}

\begin{definition}
  The \textbf{scan recovery} from a sequence $\mathbf{y}$ to a sequence $\mathbf{y^\kappa}$ in the scheme $ \kappa  $ can be formulated as follows:
  \begin{align}\label{grecover}
    \mathbf{y}^\kappa=\mathbf{R}_\kappa^{-1} \mathbf{y}.
  \end{align}
\end{definition}
With the generalized scan exchanging and scan recovery approach, we can define the GSSM module.
\begin{definition}
   \textbf{The GSSM module} maps a sequence $\mathbf{z} \in \mathbb{R}^T$ to a sequence $\mathbf{y}^\kappa \in \mathbb{R}^T$ with scheme $\kappa$ and SSM operator $SSM(\mathbf{A},\mathbf{B},\mathbf{C})$, denoted as $\mathbf{y}^\kappa=GSSM(\mathbf{A},\mathbf{B},\mathbf{C},\kappa)(\mathbf{z})$. The module can be formulated as
   \begin{align}
    &\mathbf{x^\kappa}=\mathbf{R}_\kappa \mathbf{z},\label{GSSM-SSM}\\
    &\mathbf{y}=SSM(\mathbf{A},\mathbf{B},\mathbf{C})(\mathbf{x}^\kappa),\\
    &\mathbf{y}^\kappa=\mathbf{R}_\kappa^{-1} \mathbf{y}.\label{defi5-3}
  \end{align}
  \end{definition}

Therefore, we have the following proposition for GSSM.
%The mapping of the GSSM module can be expressed as a matrix transformation.
% \begin{remark}
%   Through the generalized scan expanding approach, the operations in the VSSM module can be summarized.
%   In the VSSM module, the 2D feauture first expand as 1D sequence through the generalized scan expanding approach and the addressed by the SSM operator. The output is recovered through an inverse process of the expanding.
% \end{remark}

%With the generalized scan expanding approach and the clear operations in the VSSM module, we can prove the influence of the scan schemes on the VSSM module.
\begin{proposition}\label{VSSM-0}
  The mapping of the GSSM module with $\mathbf{h}_0=\mathbf{0}$ can be represented as a matrix transformation as following
  \begin{equation}
    \mathbf{y}^\kappa=GSSM(\mathbf{A},\mathbf{B},\mathbf{C},\kappa)(\mathbf{z})=\mathbf{U}^\kappa \mathbf{z}\label{GSSM},
  \end{equation}
 where we have $\mathbf{U}^\kappa=\mathbf{R}_\kappa^{-1}\mathbf{M}\mathbf{R}_\kappa$.
\end{proposition}
\begin{proof}
 According to (\ref{SSM-0}), the SSM operator in the GSSM can be expressed as 
  \begin{align}
    \mathbf{y}=SSM(\mathbf{A},\mathbf{B},\mathbf{C})(\mathbf{x}^\kappa)=\mathbf{Mx}^\kappa \label{prop2-prof}.
  \end{align}
Substituting (\ref{prop2-prof}) and (\ref{GSSM-SSM}) into (\ref{defi5-3}), we obtain
  \begin{align}\label{y^kappa}
    \mathbf{y}^\kappa=\mathbf{R}_\kappa^{-1}\mathbf{y}=\mathbf{R}_\kappa^{-1}(\mathbf{Mx}^\kappa)=\mathbf{R}_\kappa^{-1}\mathbf{M}\mathbf{R}_\kappa \mathbf{z}.
  \end{align}
  % The generalized scan expanding operation and the SSM operator follows the expression in (\ref{expand}) and (\ref{SSMop}). They can be unified by expressing the SSM operator as matrix transformation as following
  % \begin{align}
  %   \mathbf{y}=SSM(\mathbf{A},\mathbf{B},\mathbf{C})(\mathbf{x})=\mathbf{Mx}=\mathbf{M}(\mathbf{R}_\kappa \mathbf{z})
  % \end{align}
  % the vector $\mathbf{y}$ is then recover by the inverse of the expanding process, which can be expressed as
  % \begin{align}
  %   \mathbf{y}^\kappa=\mathbf{R}_\kappa^{-1}\mathbf{y}
  % \end{align}
  % we formulate the three processes in one expression as

  % From the expression, it can be discovered that the according to the associative property of matrix multiplication, the scan and recovery process applied on the sequence can be seen as a similarity transformation applied on the SSM matrix $\mathbf{M}$.
Let $\mathbf{U}^\kappa$ be defined as $\mathbf{R}_\kappa^{-1}\mathbf{M}\mathbf{R}_\kappa$. Thus we obtain (\ref{GSSM}).
\end{proof}

\begin{figure*}[t]
  \begin{center}
    \includegraphics[width=0.99\textwidth]{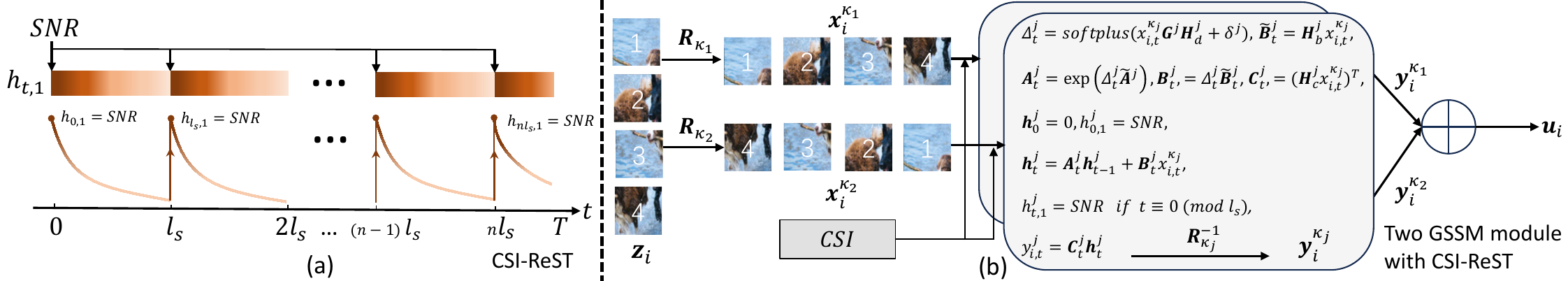}
  \end{center}
    \caption{{(a): The concept of the CSI-ReST method. (b): The computational flow of the two GSSM modules with CSI-ReST.}}
    \label{V-S6}
    \vspace{-0.3 cm}
\end{figure*}

Proposition 2 demonstrates that the scan exchanging and scan recovery operations in the GSSM module preserve the matrix transformation property. These operations merely convert the transformation matrix into its similar matrix based on the scan scheme $\kappa$. The similarity transformations affect the perceptible elements in the GSSM module. For image transmission, it is crucial to capture the global elements of the vectorization image feature $\mathbf{z}$ efficiently. To achieve this, we design the use of two GSSM modules: one applies scan scheme $\kappa_1$, where $\mathbf{R}_{\kappa_1}$ is the identity $\mathbf{I}$, and the other adopts the scan scheme $\kappa_2$, where the elements in row $i$ and column $j$ of matrix $\mathbf{R}_{\kappa_2}$ is $1$ if $i+j=T+1$. The outputs of both GSSM modules are then combined to produce the final output $\mathbf{u}$, which can be derived as follows:
\begin{align}\label{DGSSM}
  \mathbf{u}&= \mathbf{y}^ {\kappa_1}+\mathbf{y}^ {\kappa_2}=GSSM(\mathbf{A}^1,\mathbf{B}^1,\mathbf{C}^1,\kappa_1)(\mathbf{z})\nonumber\\
  &+GSSM(\mathbf{A}^2,\mathbf{B}^2,\mathbf{C}^2,\kappa_2)(\mathbf{z})
  =\mathbf{U}^{\kappa_1} \mathbf{z}+\mathbf{U}^{\kappa_2} \mathbf{z}.% \mathbf{z}%\mathbf{M^{1}} \mathbf{z}+\mathbf{R}_{\chi_2}^{-1}\mathbf{M^{2}}\mathbf{R}_{\chi_2} \mathbf{z}
\end{align}
%where $\mathbf{M^{1}} \mathbf{z}=SSM(\mathbf{A^1},\mathbf{B^1},\mathbf{C^1})(\mathbf{z})$ and $\mathbf{M^{2}} \mathbf{z}=SSM(\mathbf{A^2},\mathbf{B^2},\mathbf{C^2})(\mathbf{z})$
\begin{proposition}\label{pro4}
  The design in (\ref{DGSSM}) is sufficient to capture global information. 
\end{proposition}
\begin{proof}
  We can write the matrices $\mathbf{U}^{\kappa_1}$ and $\mathbf{U}^{\kappa_2}$ in element wise. The elements in row $i$ and column $j$ of $\mathbf{U}^{\kappa_1}$ is
  \begin{align}
    U^{\kappa_1}_{ij}=\mathbf{C}^1_i\mathbf{A}^{1 \times }_{i:j}\mathbf{B}^1_j,
  \end{align}
and according to the definition of $\mathbf{R}_{\kappa_2}$, $\mathbf{U}^{\kappa_2}$ is obtained by reversing the arrangement in both row and column directions. The element in $\mathbf{U}^{\kappa_2}$ can be derived as
\begin{align}
U^{\kappa_2}_{ij}=\mathbf{C}^2_{T+1-i} \mathbf{A}^{2 \times }_{(T+1-i):(T+1-j)}\mathbf{B}^2_{T+1-j}.
\end{align}
%   \begin{equation}
%   m=T+1-i, \ n=T+1-j
% \end{equation}
According to (\ref{DGSSM}), the $t$-th element of $\mathbf{u}$ is
\begin{align}\label{global}
  u_t=&\sum_{s = 1}^{t}\mathbf{C}^1_t \mathbf{A}_{t:s}^{1\times} \mathbf{B}^1_s z_s+\nonumber\\
  &\sum_{s = t}^{T}\mathbf{C}^2_{T+1-t} \mathbf{A}_{(T+1-t):(T+1-s)}^{2\times} \mathbf{B}^2_{T+1-s} z_s.
\end{align}
According to (\ref{Atimes}), $\mathbf{A}_{t:s}^{1\times}$ consists of learnable non-zero parameters when $s$ is in the summarization ranging from $1$ to $t$. Similarly, $\mathbf{A}_{(T+1-t):(T+1-s)}^{2\times}$ contains learnable non-zero parameters when $s$ is in the summarization ranging from $t$ to $T$. Therefore, for any element $\mathbf{u}_t$ in the output $\mathbf{u}$, it can combine all input elements with learnable parameters, demonstrating that the designed module, expressed as (\ref{DGSSM}), is capable of capturing global information.
\end{proof}

\subsection{CSI as Residual States for Channel Adaptation}
%In this subsection, we proposed our CSI-ReST method, achieving effective channel adaptation while maintaining zero parameters and zero computational overhead.
To design the customized CSI-ReST method, we further explore the structural characteristics of the GSSM model and uncover other exploitable properties. Tracing back to its origins, the SSM model is built upon the continuous-time system state equation, as shown in (\ref{LTI}). Mathematically, the solution to (\ref{LTI}) can be expressed as follows:
\begin{align}
\mathbf{h}(t)=e^{\mathbf{\tilde{A}}(t-t_0)}\mathbf{h}(t_0)+\int_{t_0}^{t}e^{\mathbf{\tilde{A}}(t-\tau)}\mathbf{\tilde{B}}x(\tau)d\tau,\\
y(t)=\underbrace{\mathbf{\tilde{C}}e^{\mathbf{\tilde{A}}(t-t_0)}\mathbf{h}(t_0)}_{zero-input\ response}+\underbrace{\int_{t_0}^{t}\mathbf{\tilde{C}}e^{\mathbf{\tilde{A}}(t-\tau)}\mathbf{\tilde{B}}x(\tau)d\tau}_{Zero-state\ response}.
\end{align}

It can be observed that the output consists of both the zero-state response and the zero-input response. In conventional SSM models, only the zero-state response contributes to the output, even though both responses are calculated, leading to the zero-input response being underutilized. By analyzing the impact of the zero-input response on the GSSM model, we can explore how to leverage this otherwise wasted intrinsic feature for integrating CSI.
\begin{proposition}\label{pro5}
  The mapping of the GSSM module with any zero state $\mathbf{h}_0$ can be represented as two independent matrix transformations applied to the input and the zero state as follows:
  \begin{align}\label{CSI-GSSM}
    \mathbf{y}^\kappa=GSSM(\mathbf{A},\mathbf{B},\mathbf{C},\kappa)(\mathbf{z},\mathbf{h}_0)=\mathbf{U}^\kappa \mathbf{z}+\mathbf{V}^{\kappa} \mathbf{h}_0,
    % &\mathbf{U}^\kappa=\mathbf{R}_\kappa^{-1}\mathbf{M}\mathbf{R}_\kappa\\
    % &M_{ij}:=\mathbf{C}_i\mathbf{A}^\times_{i:j}\mathbf{B}_j\nonumber
  \end{align}
  where $\mathbf{V}^{\kappa}=\mathbf{R}_{\kappa}^{-1}diag(\mathbf{C}_1\mathbf{A}^{\times}_{1:0},\mathbf{C}_2\mathbf{A}^{\times}_{2:0},...,\mathbf{C}_T\mathbf{A}^{\times}_{T:0})$.
\end{proposition}
\begin{proof}
  According to (\ref{GSSM}), we expand the SSM operator in the GSSM module from its recurrent form to a matrix operations form. The hidden state $\mathbf{h}_t$ with input $\mathbf{x}^\kappa$ for all $t$ can be expressed as:
  \begin{align}
    \mathbf{h}_1&=\mathbf{A}_1\mathbf{h}_0+\mathbf{B}_1x^\kappa_1\nonumber,\\
    \mathbf{h}_2&=\mathbf{A}_2\mathbf{h}_1+\mathbf{B}_2x^\kappa_2=\mathbf{A}_2\mathbf{A}_1\mathbf{h}_0+\mathbf{A}_2\mathbf{B}_1x^\kappa_1+\mathbf{B}_2x^\kappa_2\nonumber,\\
    &\vdots\nonumber\\
    \mathbf{h}_{T}&=\mathbf{A}^\times_{T:0}\mathbf{h}_0+\sum_{s=1}^{T}\mathbf{A}^\times_{T:s}\mathbf{B}_s x^\kappa_s.
    \end{align}
    Therefore, ${y}_t$ can be written as:
  \begin{align}\label{CSISSM}
    y_t=\mathbf{C}_t\mathbf{h}_t=\mathbf{C}_t\mathbf{A}^\times_{t:0}\mathbf{h}_0+\sum_{s=1}^{t}\mathbf{C}_t\mathbf{A}^\times_{t:s}\mathbf{B}_s x^\kappa_s.
  \end{align}
  Vectorizing equation (\ref{CSISSM}) over $t \in [1,T]$, the sum of accumulations can be expressed in matrix multiplication form as the method deriving $\mathbf{M}$. The first term can be rewritten in block matrices multiplication form by expressed it as matrix multiplication with matrix $diag(\mathbf{C}_1\mathbf{A}^{\times}_{1:0},\mathbf{C}_2\mathbf{A}^{\times}_{2:0},...,\mathbf{C}_T\mathbf{A}^{\times}_{T:0})$. Thus combining (\ref{y^kappa}), we derive the expression in (\ref{CSI-GSSM}).
%   \begin{align}\label{CSI-GSSM}
%     &\mathbf{y}^\kappa=\mathbf{U}^\kappa \mathbf{z}+\mathbf{V} \mathbf{h}_0.
% \end{align}  
\end{proof}
From Proposition \ref{pro5}, we find that \textbf{the GSSM module responds to the zero state $\mathbf{h}_0$}. This insight inspired us to propose CSI-ReST, a novel channel adaptation method that leverages the inherent (i.e., endogenous intelligence) adaptive capability of the GSSM, as shown in Fig. \ref{V-S6} (a). We first inject CSI into the first element of the initial hidden state of the GSSM modules by setting $h_{0,1}=SNR$ to leverage the inherent response and adjust the output sequence based on CSI for channel adaption. 

\begin{remark}
\textbf{CSI forgetting:} Equ. (\ref{CSISSM}) also implies that the response involves the matrix $\mathbf{A}_{1:T}$, which functions as a forgetting gate (e.g. with typical values ranging from $0.95$ to $0.97$), indicating $\mathbf{A}^\times_{t:0}$ becomes smaller as $t$ increases (e.g. $6.37\times 10^{-11}$ at $t=500$). This suggests that the influence of CSI on $y_t$ may diminish as $t$ increases. This characteristic makes it difficult to effectively adjust the output based solely on the initial state.
\end{remark}

\addtolength{\topmargin}{0.05in}
\begin{algorithm}[t]
  \hspace*{0.02in} {\bf \small{Input:}}
	\small{Input sequence $\mathbf{z}_i$, The SNR of the channel $SNR$;} \\
	\hspace*{0.02in} {\bf \small{Output:}}
	\small{Output sequence $\mathbf{u}_i$;} 
  %\hspace*{0.02in} {\bf \small{Learnable Parameters:}}
  %\small{$\mathbf{G}$, $\mathbf{H}_1$, $\mathbf{H}_2$, $\mathbf{H}_3$, $\mathbf{A}$ and $\mathbf{D}$}.
	\caption{The algorithm of the two GSSM modules with CSI-ReST.}
	\label{V-S6algorithm}
	\begin{algorithmic}[1]
    \STATE $\mathbf{x}^{\kappa_1}_i=\mathbf{R}_{\kappa_1}\mathbf{z}_i,\ \mathbf{x}^{\kappa_2}_i=\mathbf{R}_{\kappa_2}\mathbf{z}_i$;
    \FOR {$j=1$ to 2}
    \STATE ${\Delta}^j_t=softplus({x}_{i,t}^{\kappa_j}\mathbf{G}^j\bm{H}^j_d+{\delta}^j),~ \mathbf{\tilde{B}}^j_t=\mathbf{H}_b^j {x}_{i,t}^{\kappa_j};$
    \STATE $\mathbf{{A}}^j_t=\exp({\Delta}^j_t \mathbf{\tilde{A}}^j),~\mathbf{{B}}^j_t={\Delta}^j_t\mathbf{\tilde{B}}^j_t,~ \mathbf{C}^j_t=(\mathbf{H}_c^j {x}_{i,t}^{\kappa_j})^T$;
    \STATE $\mathbf{h}_0^j=\mathbf{0}$, $h_{0,1}^j=SNR$;
    \FOR {$t=1$ to $T$}
    \STATE $\mathbf{h}_t^j=\mathbf{A}_t^j\mathbf{h}_{t-1}^j+\mathbf{B}_t^j{x}_{i,t}^{\kappa_j}$;
    \IF {$t\equiv 0~(mod ~ l_s)$}
    \STATE $h_{t,1}^j=SNR$;
    \ENDIF
    \STATE ${y}_{i,t}^{j}=\mathbf{C}_t^j\mathbf{h}_t^j$;
    \ENDFOR
    \ENDFOR
    \STATE $\mathbf{u}_i=\mathbf{y}_{i}^{\kappa_1}+\mathbf{y}_{i}^{\kappa_1}=\mathbf{R}_{\kappa_1}^{-1}\mathbf{y}_{i}^{1}+\mathbf{R}_{\kappa_2}^{-1}\mathbf{y}_{i}^{2}$.
    %\IF {$\alpha==0$}
    %\STATE Set $\beta=0$ and then compute loss $L=L_2$.
    %\ELSIF {$\alpha==1$}
    %\STATE Compute loss $L=L_1$.
    %\ELSE
    %\STATE Compute loss $L=L_1+\lambda L_2$.
    %\ENDIF
    %\STATE Update all the parameters to minimize $L$.
	\end{algorithmic}
  
\end{algorithm}

To address the CSI forgetting problem, inspired by ResNet, the proposed CSI-ReST method samples certain residual states from $\{\mathbf{h}_0,\mathbf{h}_1,...,\mathbf{h}_T\}$ at intervals of $l_s$ and sets their first element as $SNR$. As shown in Fig. \ref{V-S6} (a), in addition to inserting the SNR into the initial state $h_{0,1}$, CSI-ReST also inserts the SNR into \textbf{residual states} such as $h_{l_s,1}$, $h_{2l_s,1}$, and so on. The operation of the CSI-ReST method is
\begin{align}\label{CSI-ReST}
  h_{t,1}=SNR, ~~ for ~~ t=0,1,...,T,~~ t\equiv 0~(mod ~ l_s). %i=1,2,...,\left\lfloor \frac{T}{l_s} \right\rfloor  
\end{align}
Here, we use the average SNR as $SNR$ for Rayleigh fading channels. As a result, most states retain the knowledge of CSI, and those residual states that forget CSI can be \textbf{refreshed}, ensuring that the output sequence is sufficiently adjusted according to CSI for effective channel adaptation. Meanwhile, the high-dimensional hidden states of GSSM allow it to store extensive knowledge from the past sequence. Therefore, using one dimension of the residual states to store CSI has little impact on the model's ability to process sequences but is crucial for refreshing CSI knowledge. In summary, the CSI-ReST method selects specific residual states from the GSSM and injects CSI into them to facilitate channel adaptation. 

\begin{remark}
CSI-ReST utilizes the recurrent structure and high-dimension state characteristics of GSSM. The recurrent form allows the SNR to be embedded in $\mathbf{h}_0$, ensuring that the output sequence is adjusted based on CSI as described in (\ref{CSI-GSSM}). The high-dimension state space facilitates the refreshment of SNR to counteract forgetting, ensuring that all elements are sufficiently adjusted based on CSI.
\end{remark}

Benefitting from the characteristics of GSSMs, CSI-ReST not only achieves excellent channel adaptation but also operates with zero parameters and zero computational overhead. CSI-ReST consists solely of simple assignment operations, utilizing the existing computation and parameters of the GSSM modules. This process does not require any additional computation and additional parameters. Therefore, CSI-ReST achieves channel adaptation with zero parameter and computational overhead.
%In this way, we can recover the influence of CSI when its impact decays to a certain extend to ensure each element is integrated with CSI sufficiently.
\subsection{Overall Design of the VSSM-CA module}
By integrating the proposed GSSM from (\ref{GSSM}), the global information design from (\ref{DGSSM}), and the CSI-ReST channel adaptation method from (\ref{CSI-ReST}), we develop the VSSM-CA module as the backbone module for MambaJSCC.

The overall architecture of the VSSM-CA module is shown in Fig.\ref{MambaJSCC}(b). The core of the module consists of two GSSM modules with CSI-ReST. The VSSM-CA module first normalizes the input patches using layer normalization, then processes them through two branches. In the first branch, depicted in Fig.\ref{MambaJSCC}(b), an FC layer, a convolution layer, and a SiLU activation function process the patches into higher-dimensional patches. These patches are then vectorized into 1D sequences, with the $i$-th sequence denoted as $\mathbf{z}_i \in \mathbb{R}^{p_zw_z}, i=1,2,...,c_z$, where $c_z$ is the number of sequences. Each of the $c_z$ sequences is processed by two GSSM modules with CSI-ReST, and the two GSSM modules for each sequence have independent parameters. The process of the two GSSM modules with CSI-ReST is described in Fig. \ref{V-S6}(b). Following the design of the two GSSM modules in (\ref{DGSSM}) and the CSI-ReST in (\ref{CSI-ReST}), the input sequence $\mathbf{z}_i$ undergoes scan exchange with schemes $\kappa_1$ and $\kappa_2$, transforming into $\mathbf{x}^{\kappa_1}_i$ and $\mathbf{x}^{\kappa_2}_i$, as defined in (\ref{gexpand}), which are given by
\begin{align}
  \mathbf{x}^{\kappa_1}_i=\mathbf{R}_{\kappa_1}\mathbf{z}_i,\ \mathbf{x}^{\kappa_2}_i=\mathbf{R}_{\kappa_2}\mathbf{z}_i.
\end{align}
\begin{figure*}[ht]
  \centering
  \includegraphics[width=0.95\textwidth]{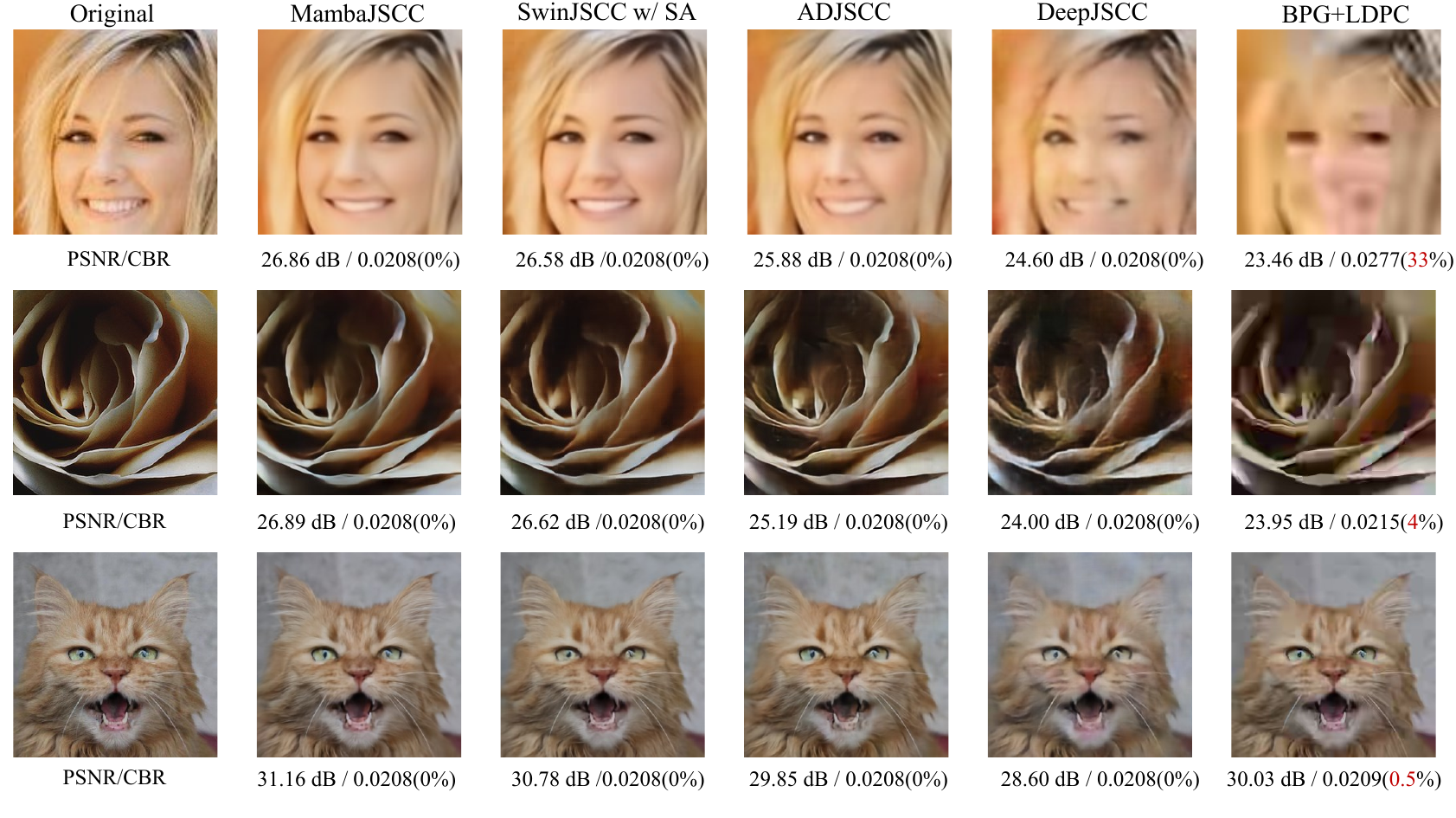}
  \caption{Examples of visual comparison between images reconstructed by different schemes at CBR=$\frac{1}{48}$. The three rows are the images coming from CelebA, DIV2K, AFHQ datasets. The red numbers indicate the percentage of channel use increased compared to MambaJSCC.}
  \label{visual}
  \vspace{-0.3 cm}
  \end{figure*}
The module then iteratively computes the output for each element through SSM and CSI-ReST. For the $t$-th element of $\mathbf{x}_i^{\kappa_j},j=1,2$, $\mathbf{h}^j_{0}=\mathbf{0}$, $h^j_{0,1}=SNR$ and the output is:
\begin{align}
  &\mathbf{h}_t^j=\mathbf{{A}}_t^j\mathbf{h}_{t-1}^j+\mathbf{{B}}_t^j{x}_{i,t}^{\kappa_j},\nonumber\\
  &h_{t,1}^j=SNR, ~~ if ~ t\equiv 0~(mod ~ l_s), \\%i=1,2,...,\left\lfloor \frac{T}{l_s} \right\rfloor  \\
  &{y}_{i,t}^{j}=\mathbf{C}_t^j\mathbf{h}_t^j\nonumber.
\end{align}

After the iterations, $\mathbf{y}_{i}^{1}$ $\mathbf{y}_{i}^{2}$ undergo scan recovery and are combined to form the final output of the two GSSM modules with CSI-ReST:
\begin{align}
  \mathbf{u}_i=\mathbf{y}_{i}^{\kappa_1}+\mathbf{y}_{i}^{\kappa_1}=\mathbf{R}_{\kappa_1}^{-1}\mathbf{y}_{i}^{1}+\mathbf{R}_{\kappa_2}^{-1}\mathbf{y}_{i}^{2}.
\end{align}

In the two GSSM modules with CSI-ReST, $\mathbf{R}_{\kappa_1}$ and $\mathbf{R}_{\kappa_2}$ are determined by the scan schemes $\kappa_1$ and $\kappa_2$. The parameters $\mathbf{{A}}^j,\mathbf{{B}}^j,\mathbf{{C}}^j$ of the $j$-th GSSM module are derived from six learnable matrices: $\mathbf{G}^j\in\mathbb{R}^{1 \times O}$, $\mathbf{H}^j_d\in \mathbb{R}^{O\times 1}$, ${\delta}^j\in\mathbb{R}^{1}$, $\mathbf{H}_b^j$, $\mathbf{H}_c^j\in\mathbb{R}^{N\times 1}$, and $\mathbf{\tilde{A}}^j \in \mathbb{R}^{N\times N}$, as follows:
\begin{align}
  &{\Delta}^j_t=softplus({x}_{i,t}^{\kappa_j}\mathbf{G}^j\bm{H}^j_d+{\delta}^j),~ \mathbf{\tilde{B}}^j_t=\mathbf{H}_b^j {x}_{i,t}^{\kappa_j},\nonumber\\
  &\mathbf{{A}}^j_t=\exp({\Delta}^j_t \mathbf{\tilde{A}}^j),~\mathbf{{B}}^j_t={\Delta}^j_t\mathbf{\tilde{B}}^j_t,~ \mathbf{C}^j_t=(\mathbf{H}_c^j {x}_{i,t}^{\kappa_j})^T,
\end{align}
where $softplus(\cdot)$ is an activation function and $O$ is their parameter generation space dimension. Therefore, the two GSSM modules with CSI-ReST independently output $c_z$ sequences. The algorithm for the two GSSM modules with CSI-ReST is summarized in Algorithm \ref{V-S6algorithm}.

The $c_z$ sequences are combined with the input using a residual connection, with the diagonal matrix $\mathbf{D}_1+\mathbf{D}_2$ applied to the input, and then reshaped as the output of the first branch. In parallel, the normalized input patches are processed by an FC layer and a SiLU activation function in the second branch. The output of the second branch is then element-wise multiplied by the output of the first branch, and the result is projected back to the original dimension using another FC layer. A residual connection adds this result to the original input. Since the two GSSM modules process the $c_z$ sequences independently, the relationship between these sequences is captured. To address this, a channel-wise multi-layer perceptron (MLP) with a residual connection is applied to the final output to merge the different sequences, producing the final output of the VSSM-CA module.
\section{EXPERIMENTAL RESULTS}\label{IV}

%In this section, we outline the detailed experimental setup and subsequently present numerous experimental results compared with a series of SwinJSCC models.
\subsection{Experiment Setup}
\begin{figure*}[htbp]
  \centering
  \subfigure[]{\label{CelebAa}\includegraphics[width=0.325\textwidth]{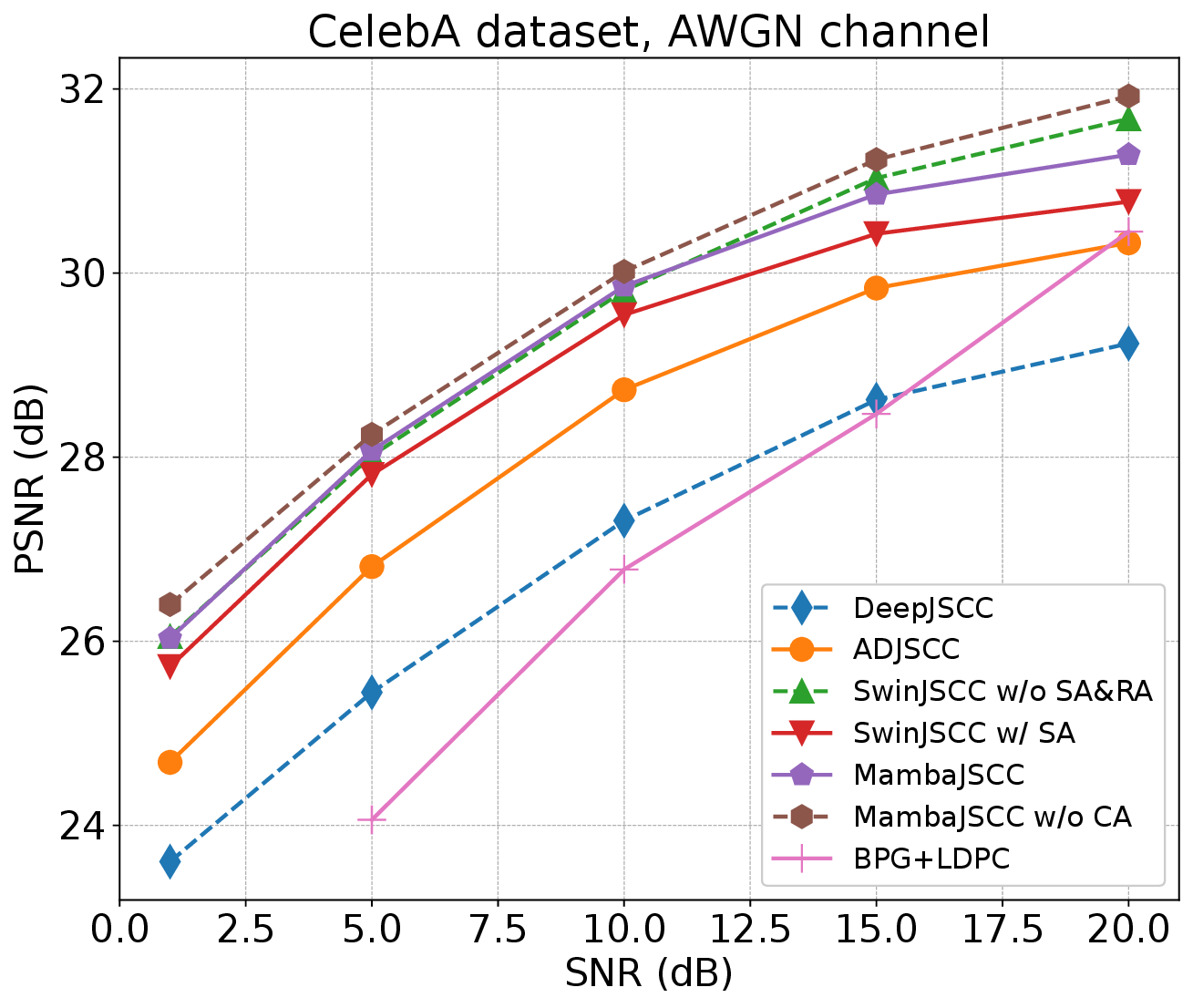}}
  \subfigure[]{\label{CelebAb}\includegraphics[width=0.325\textwidth]{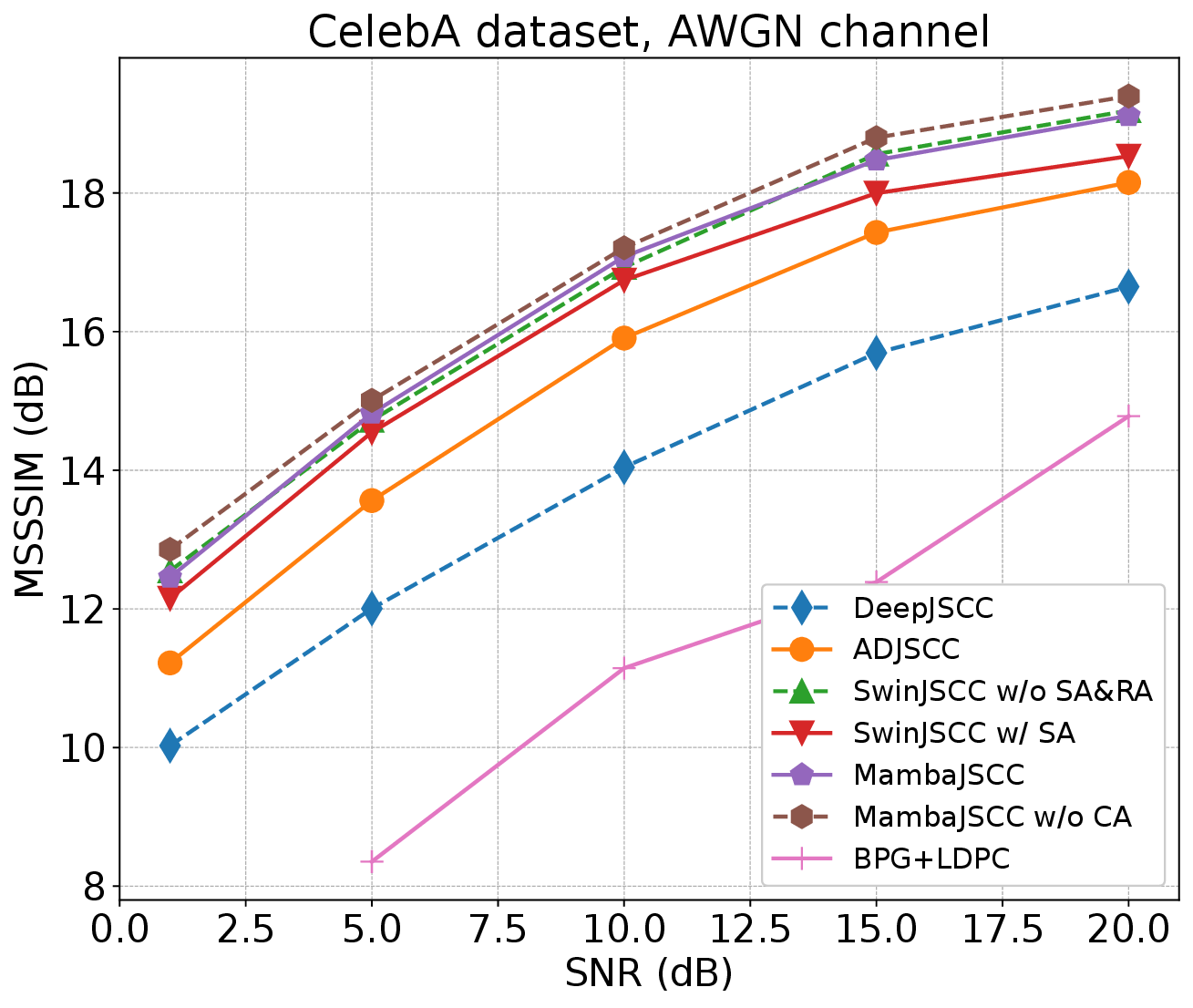}}
  \subfigure[]{\label{CelebAc}\includegraphics[width=0.334\textwidth]{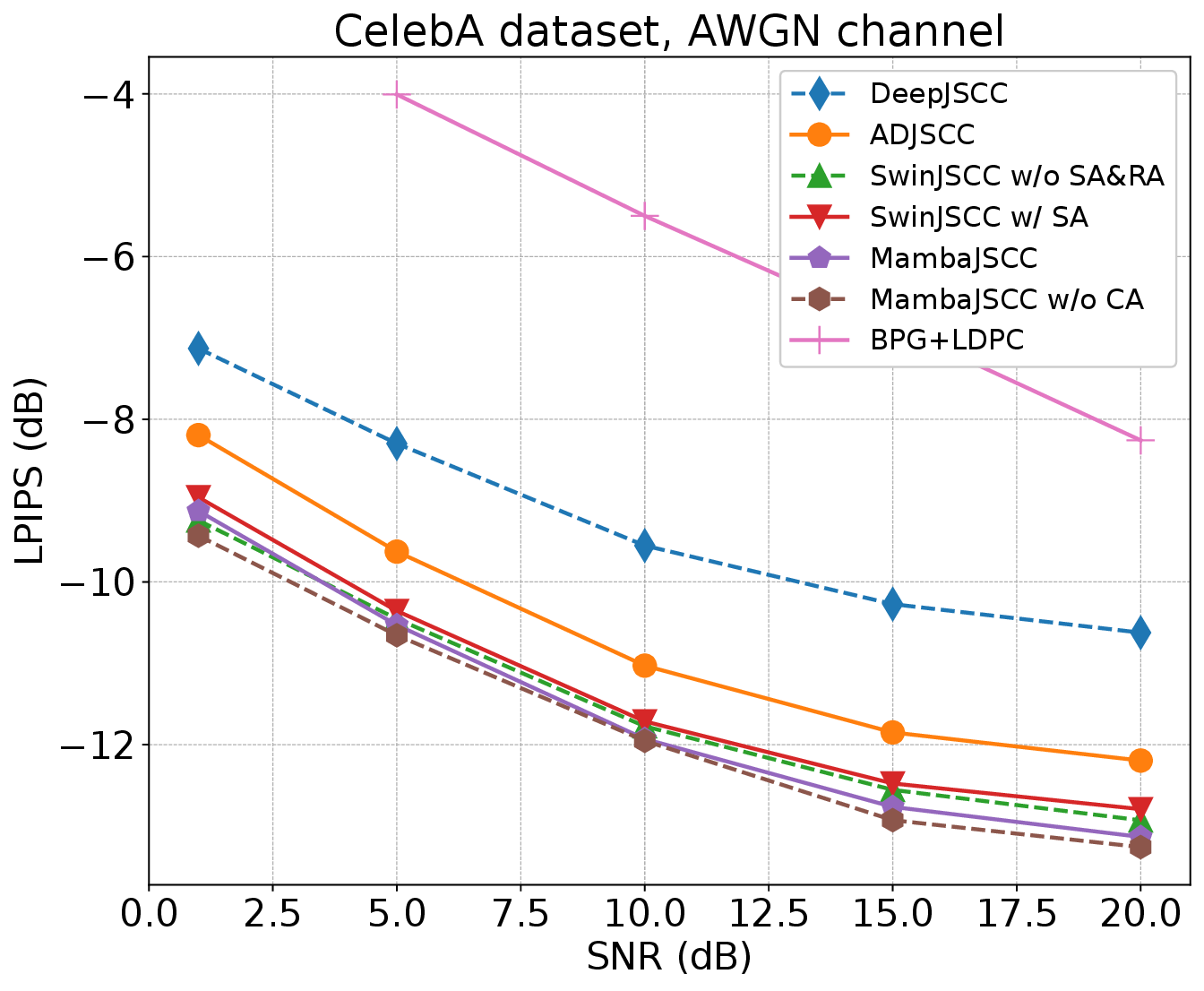}}
  \subfigure[]{\label{CelebAd}\includegraphics[width=0.325\textwidth]{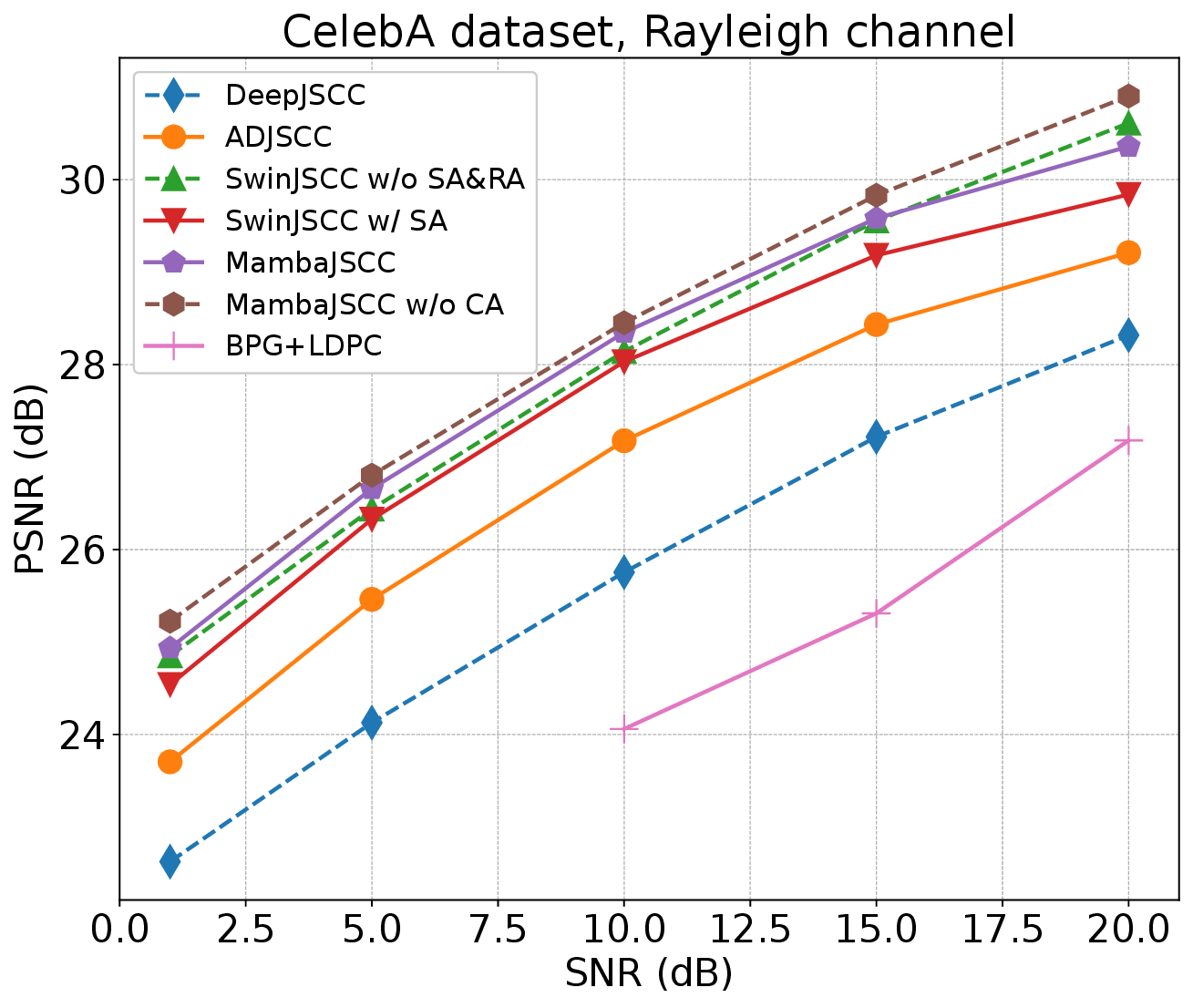}}
  \subfigure[]{\label{CelebAe}\includegraphics[width=0.325\textwidth]{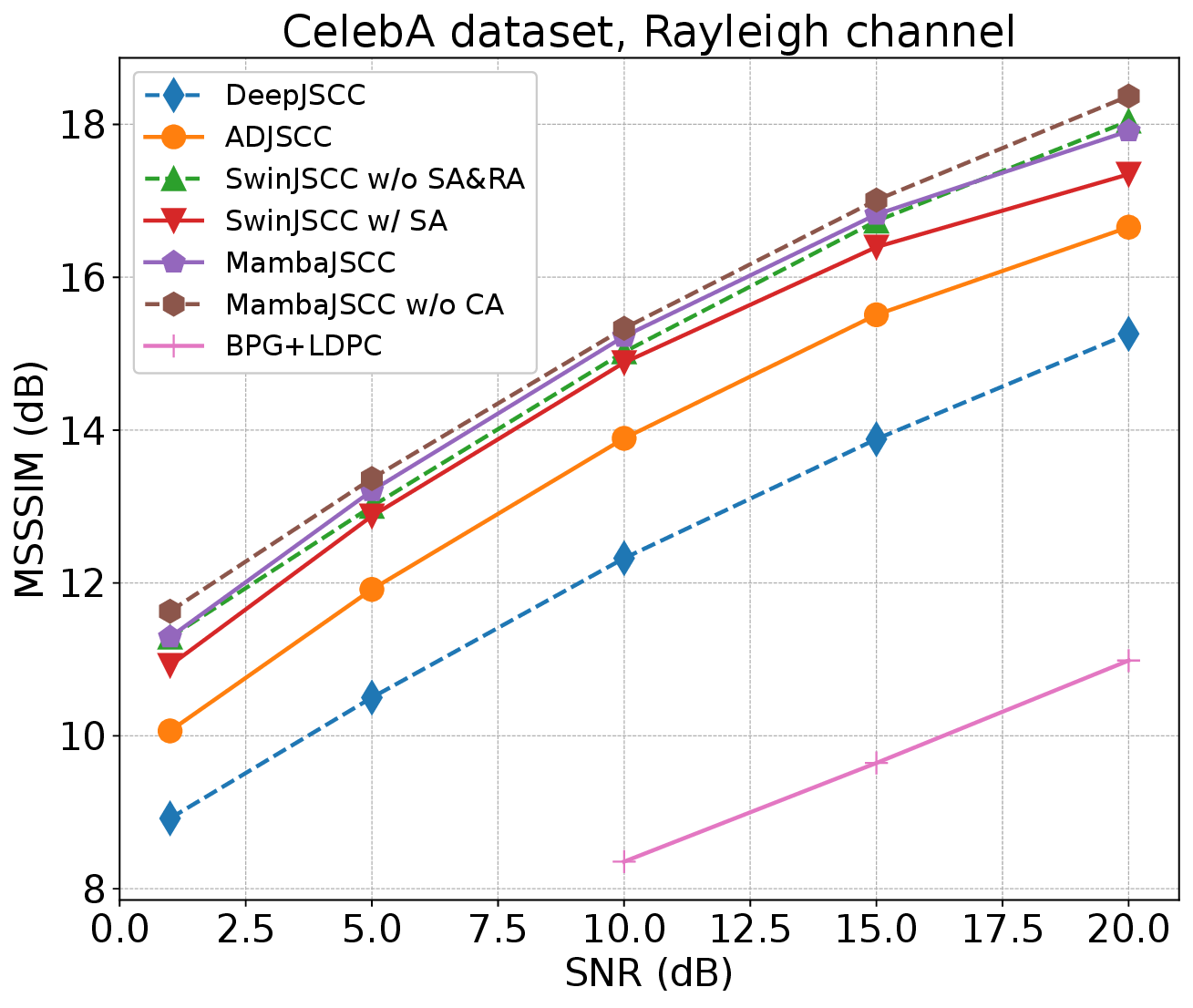}}
  \subfigure[]{\label{CelebAf}\includegraphics[width=0.334\textwidth]{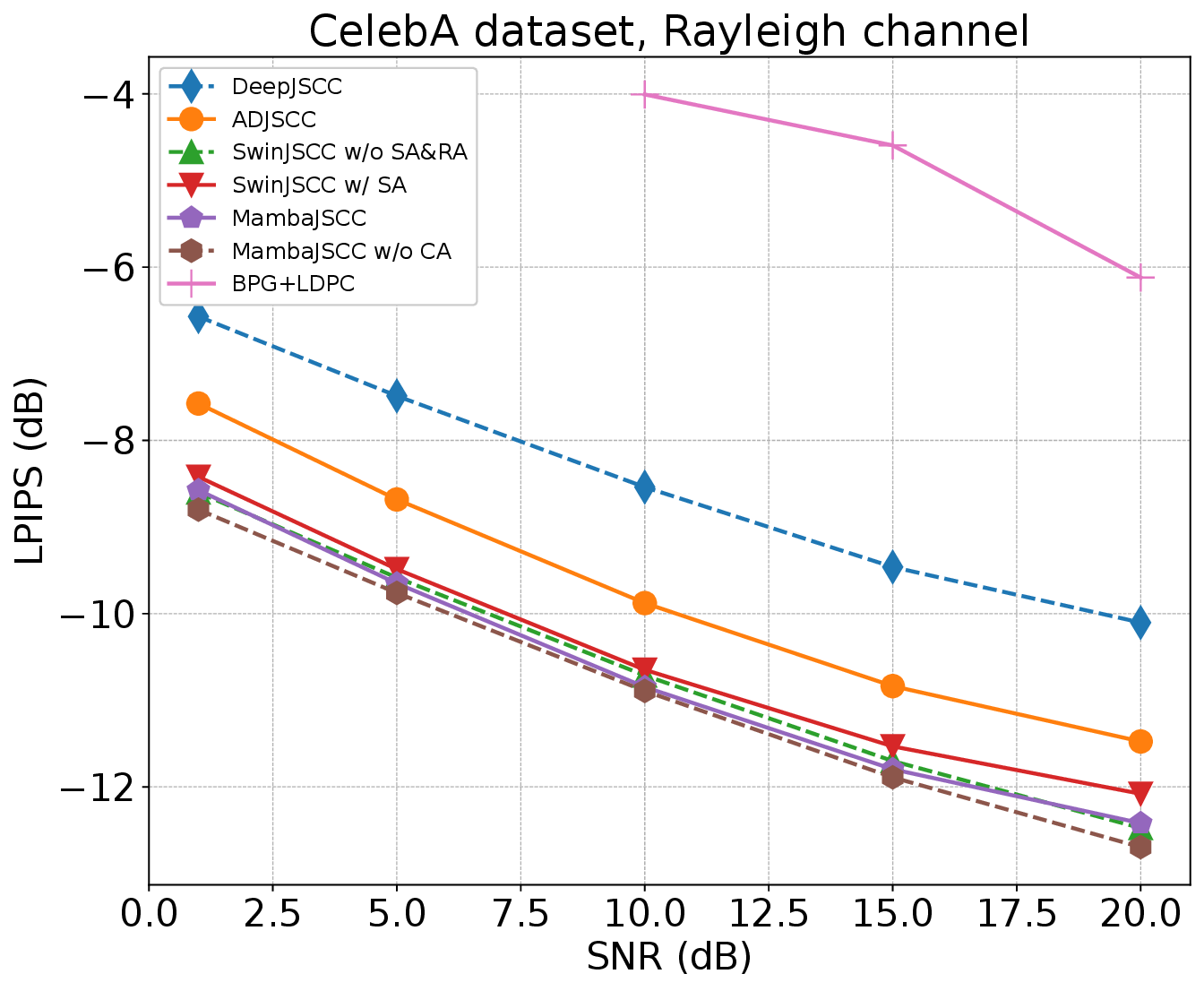}}
  \caption{(a)$\sim$(c) PSNR, MSSSIM and LPIPS performance of different models versus the SNR under the AWGN channel of CelebA dataset. (e)$\sim$(f) PSNR, MSSSIM and LPIPS performance versus the SNR under the Rayleigh fading channel of CelebA dataset. The CBR is $\frac{1}{48}$.}
  \label{CelebA}
  \vspace{-0.3 cm}
  \end{figure*}

  \begin{figure*}[htbp]
    \centering
    \subfigure[]{\label{DIV2Ka}\includegraphics[width=0.325\textwidth]{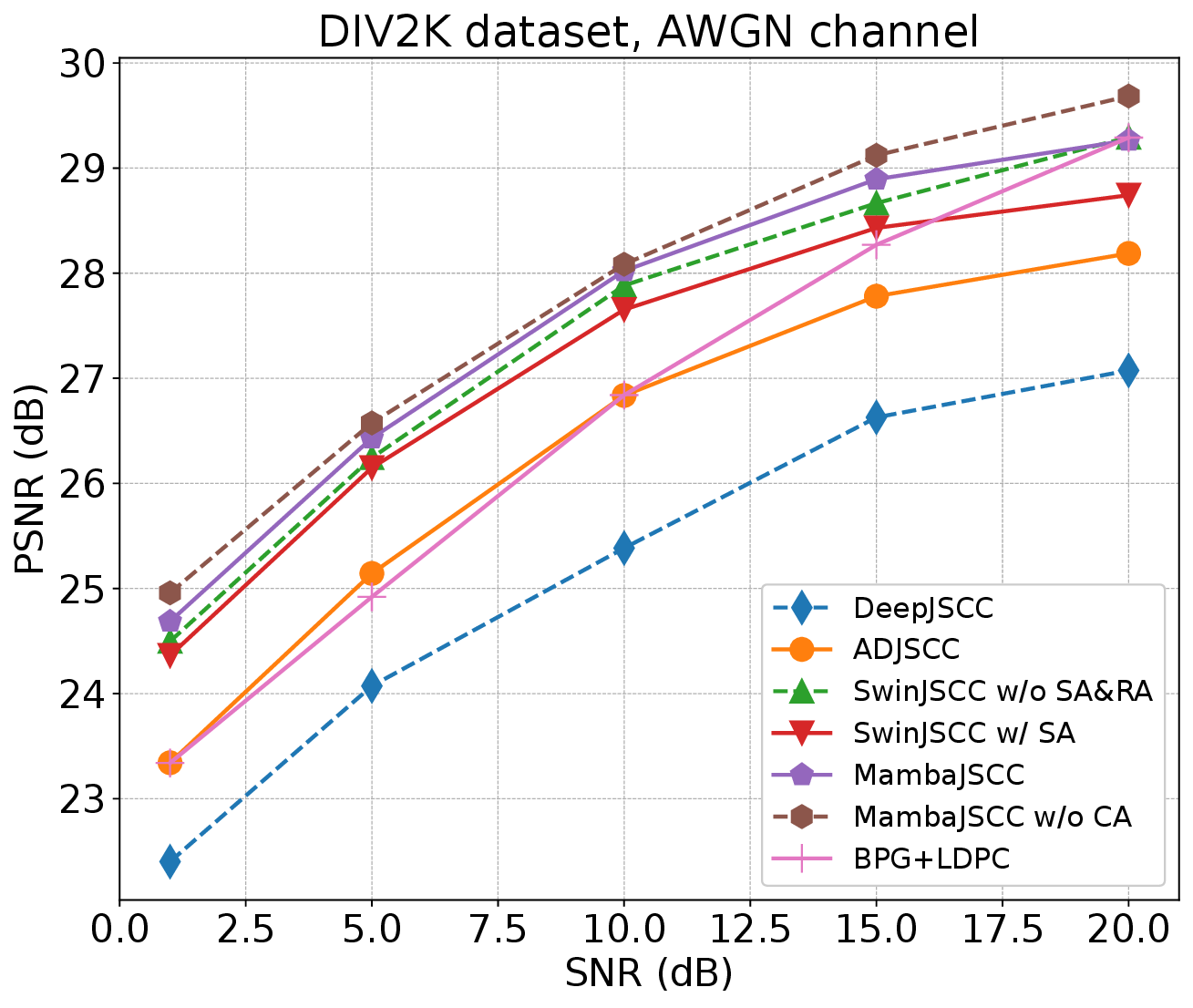}}
    \subfigure[]{\label{DIV2Kb}\includegraphics[width=0.325\textwidth]{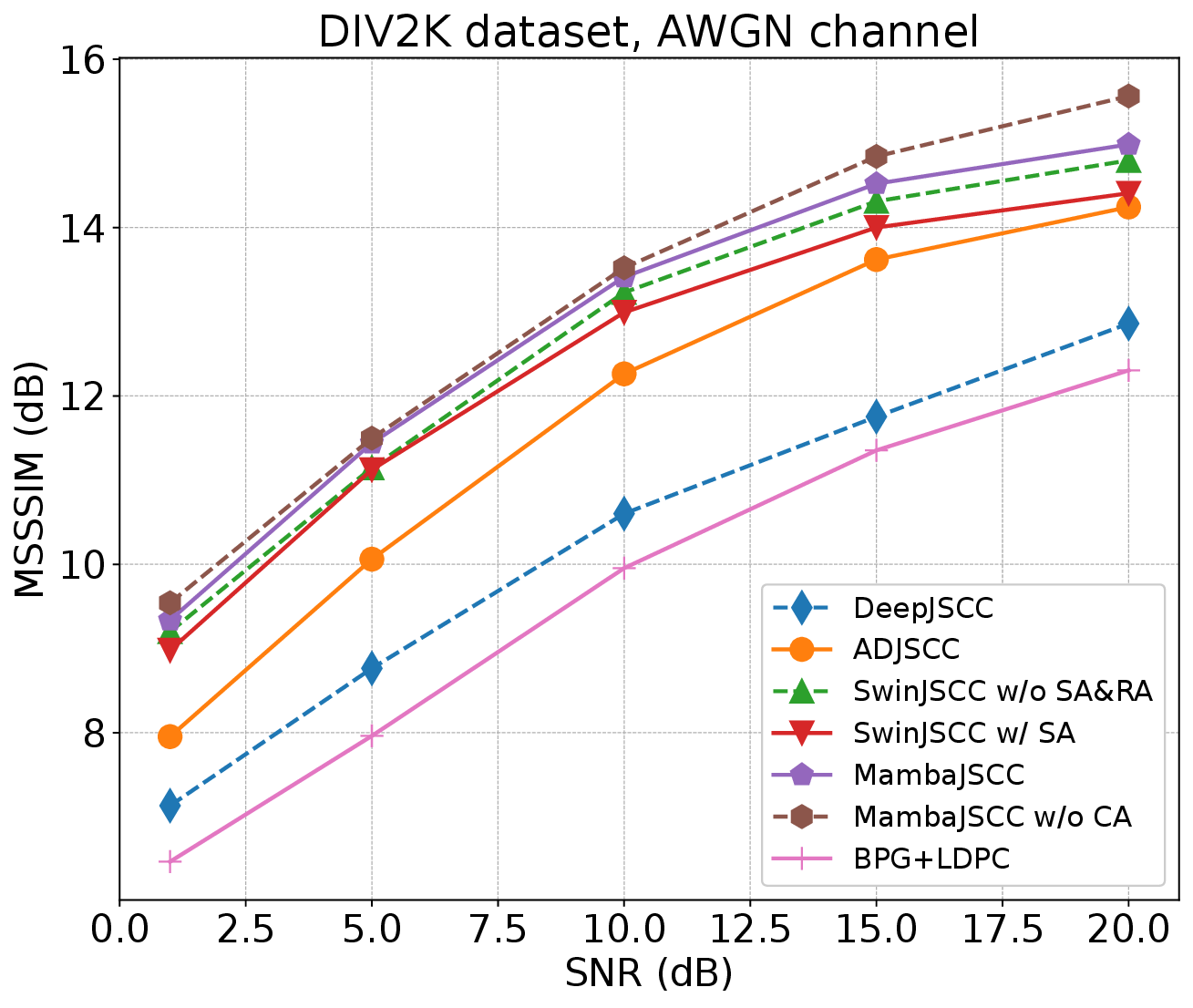}}
    \subfigure[]{\label{DIV2Kc}\includegraphics[width=0.334\textwidth]{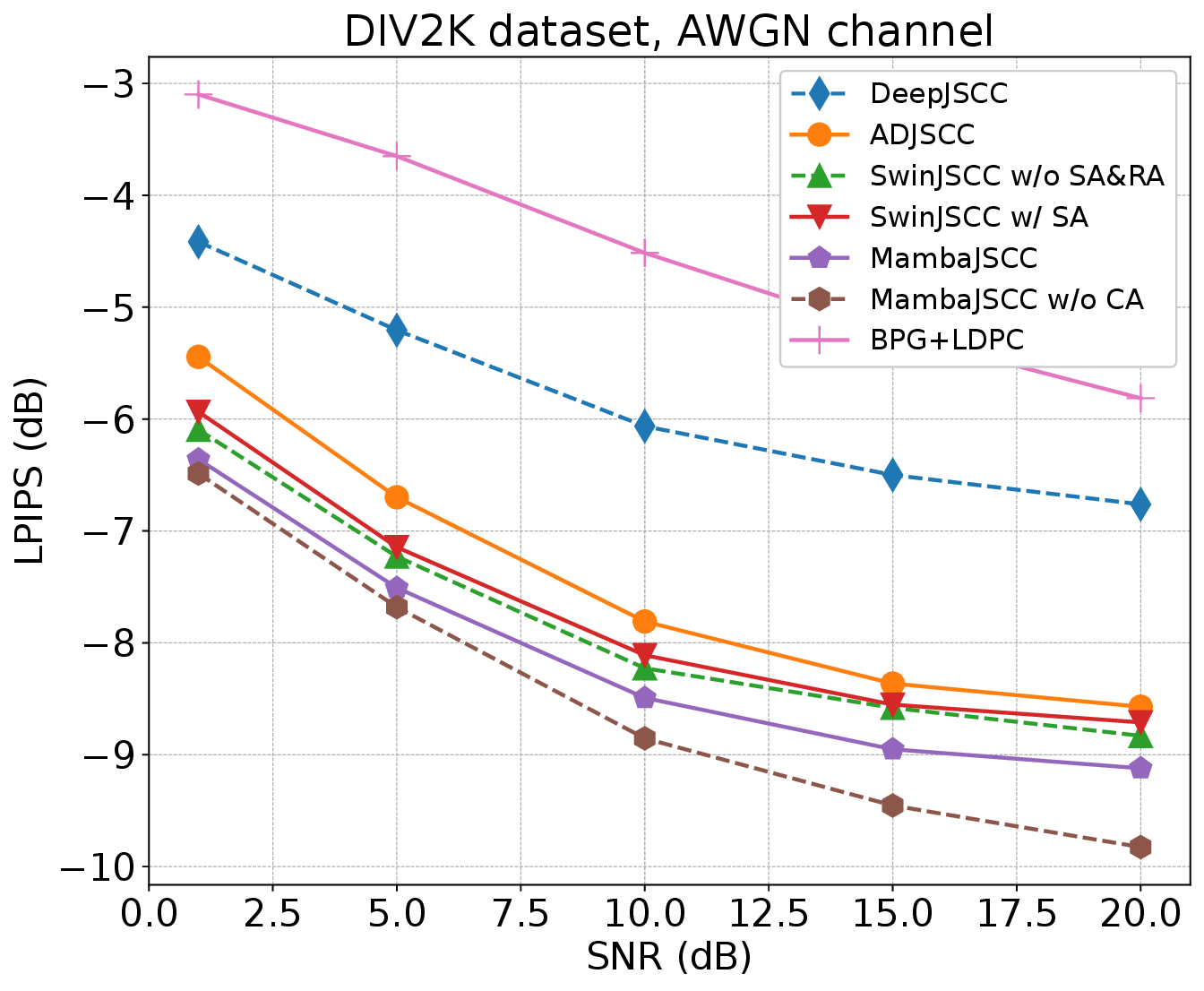}}
    \subfigure[]{\label{DIV2Kd}\includegraphics[width=0.325\textwidth]{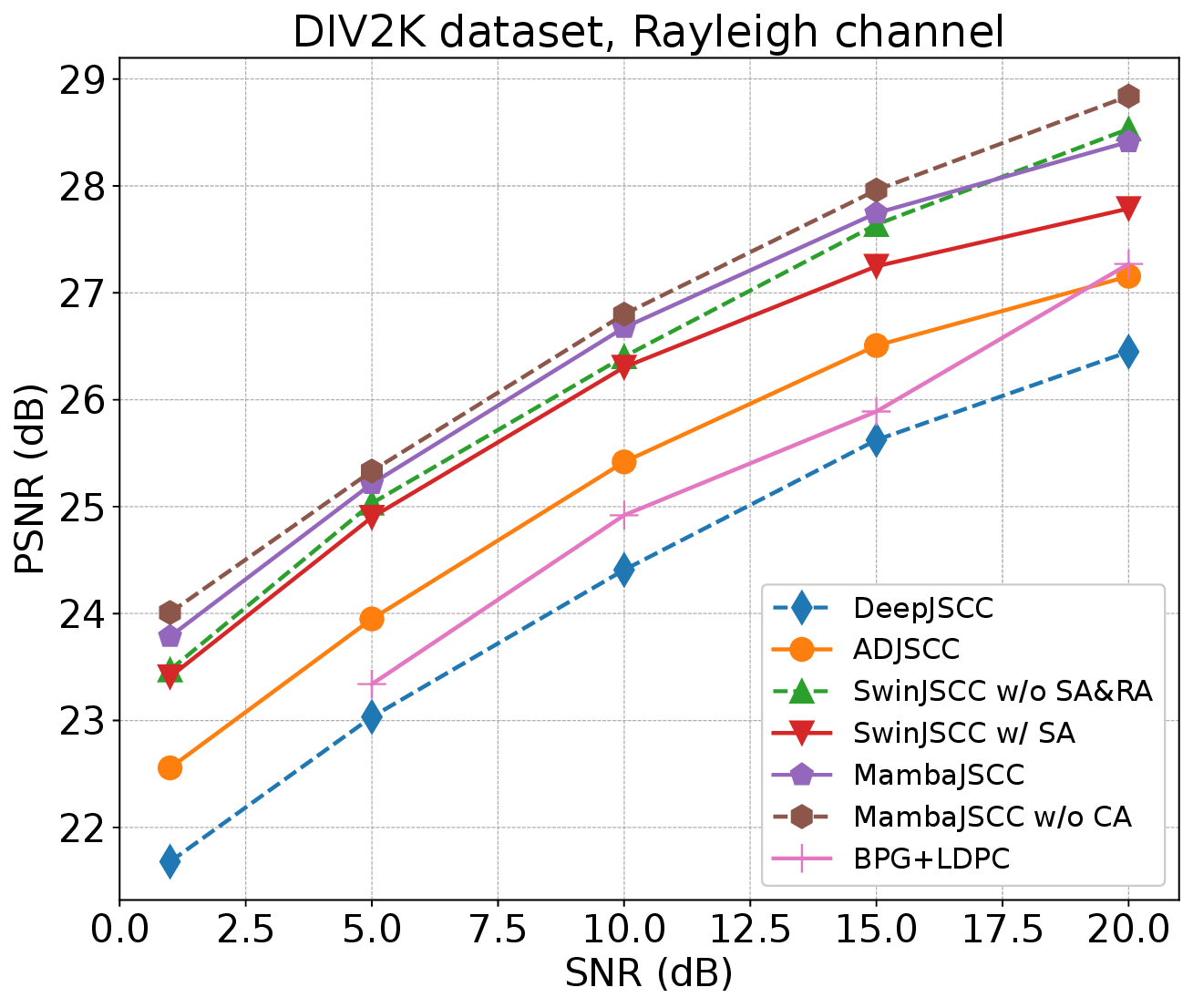}}
    \subfigure[]{\label{DIV2Ke}\includegraphics[width=0.325\textwidth]{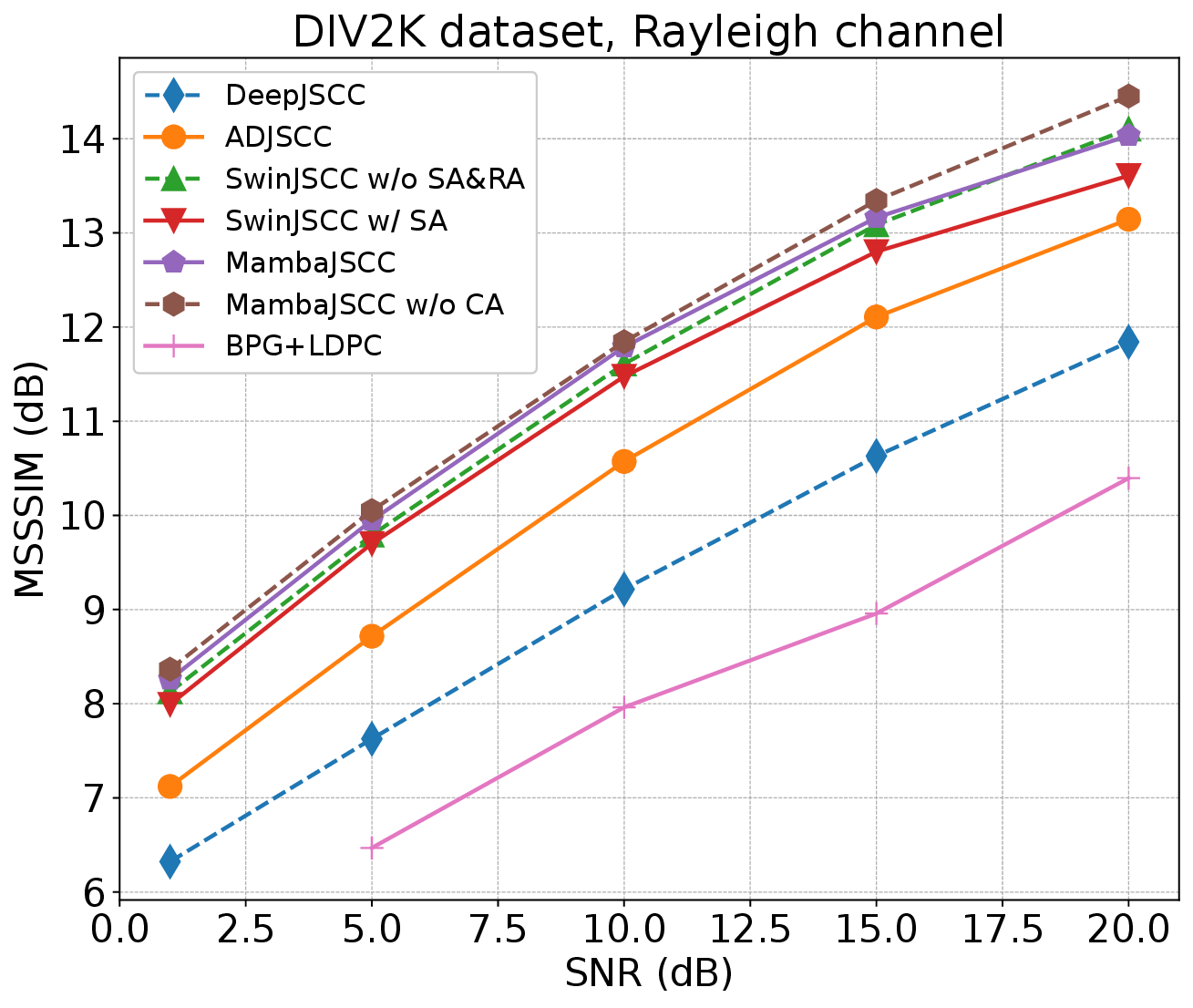}}
    \subfigure[]{\label{DIV2Kf}\includegraphics[width=0.332\textwidth]{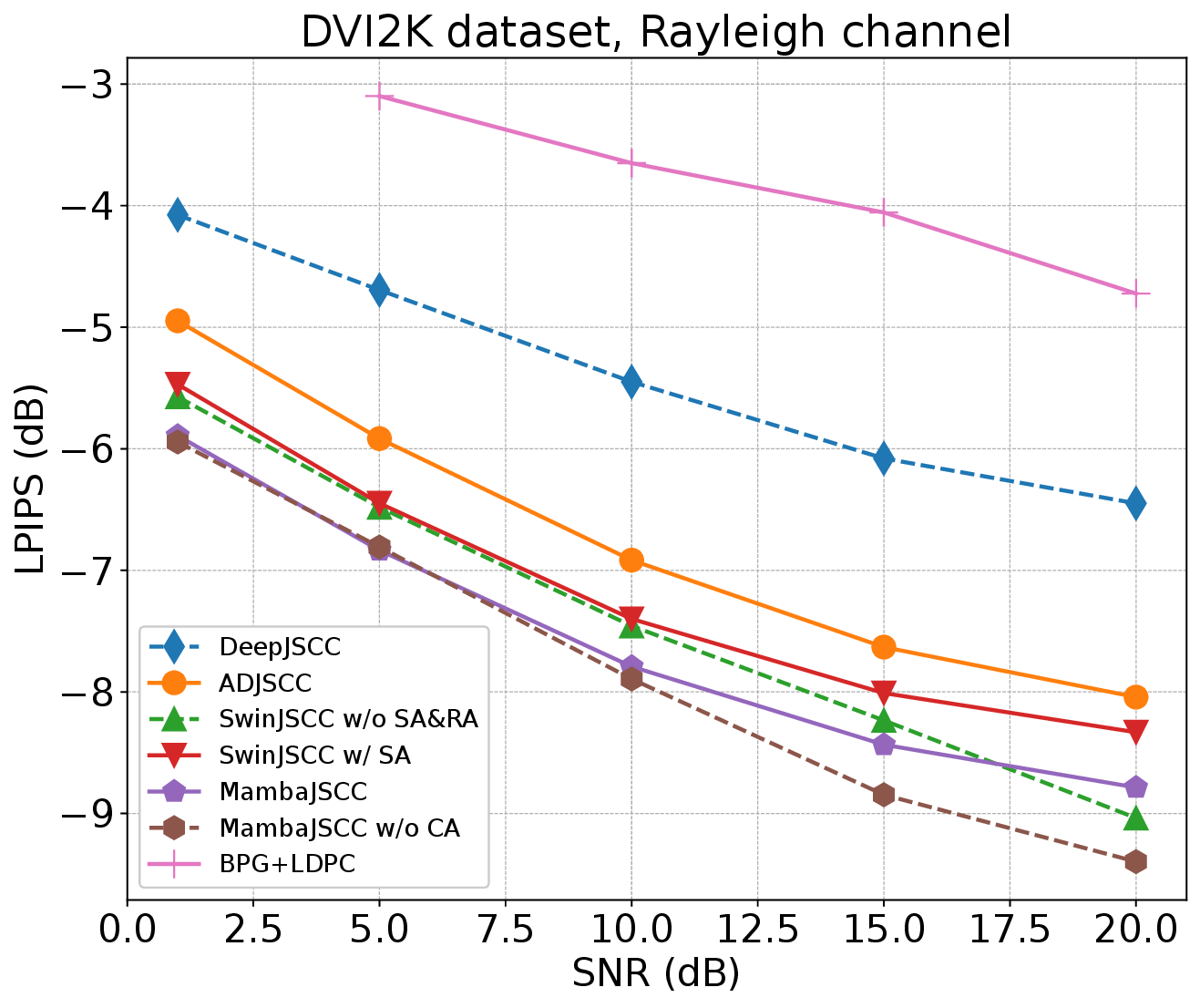}}
    \caption{(a)$\sim$(c) PSNR, MSSSIM and LPIPS performance of different models versus the SNR under the AWGN channel of DIV2K dataset. (e)$\sim$(f) PSNR, MSSSIM and LPIPS performance versus the SNR under the Rayleigh fading channel of DIV2K dataset. The CBR is $\frac{1}{48}$.}
    \label{DIV2K}
    \vspace{-0.3 cm}
    \end{figure*}

  \begin{figure*}[htbp]
    \centering
    \subfigure[]{\label{AFHQa}\includegraphics[width=0.325\textwidth]{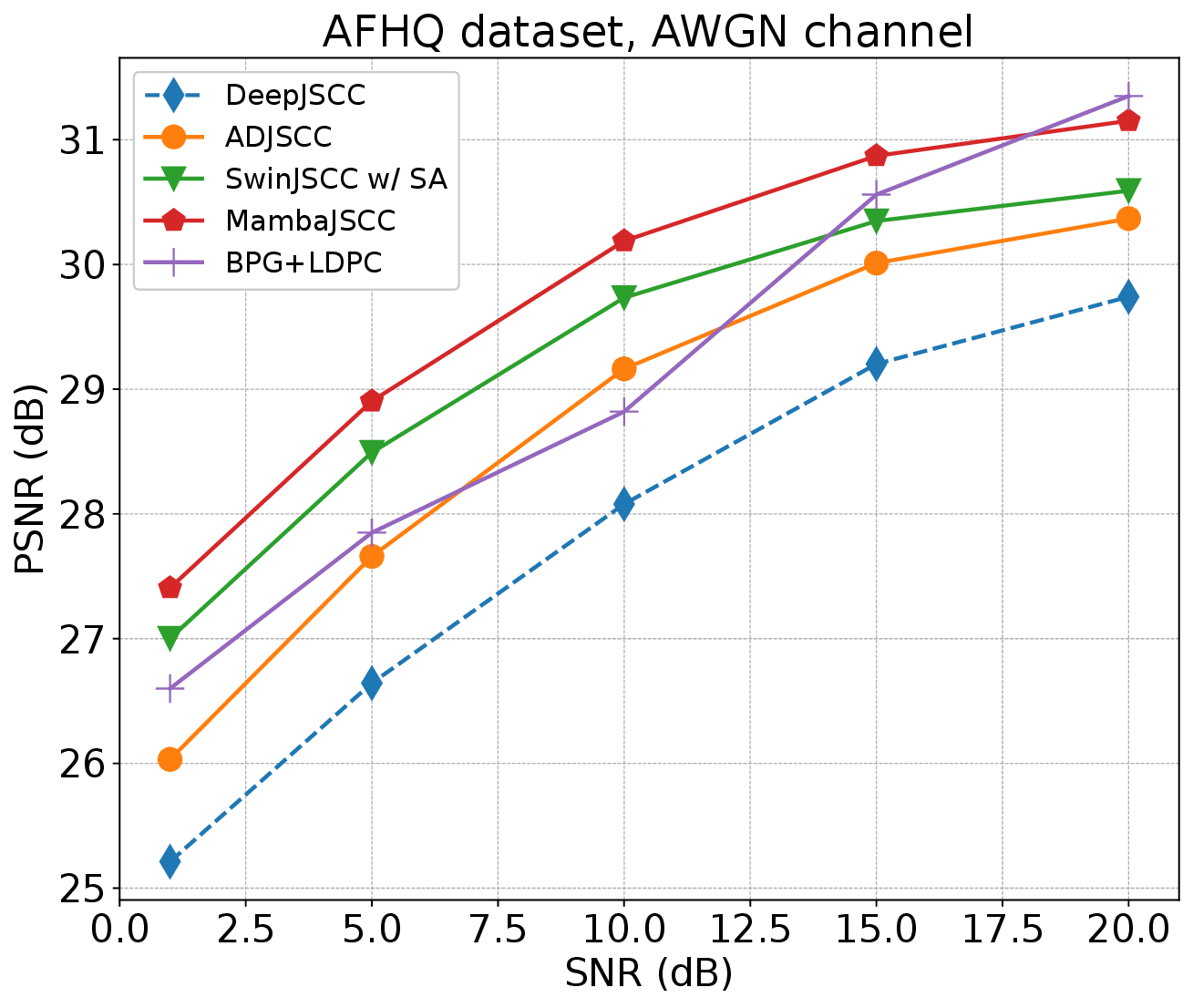}}
    \subfigure[]{\label{AFHQb}\includegraphics[width=0.325\textwidth]{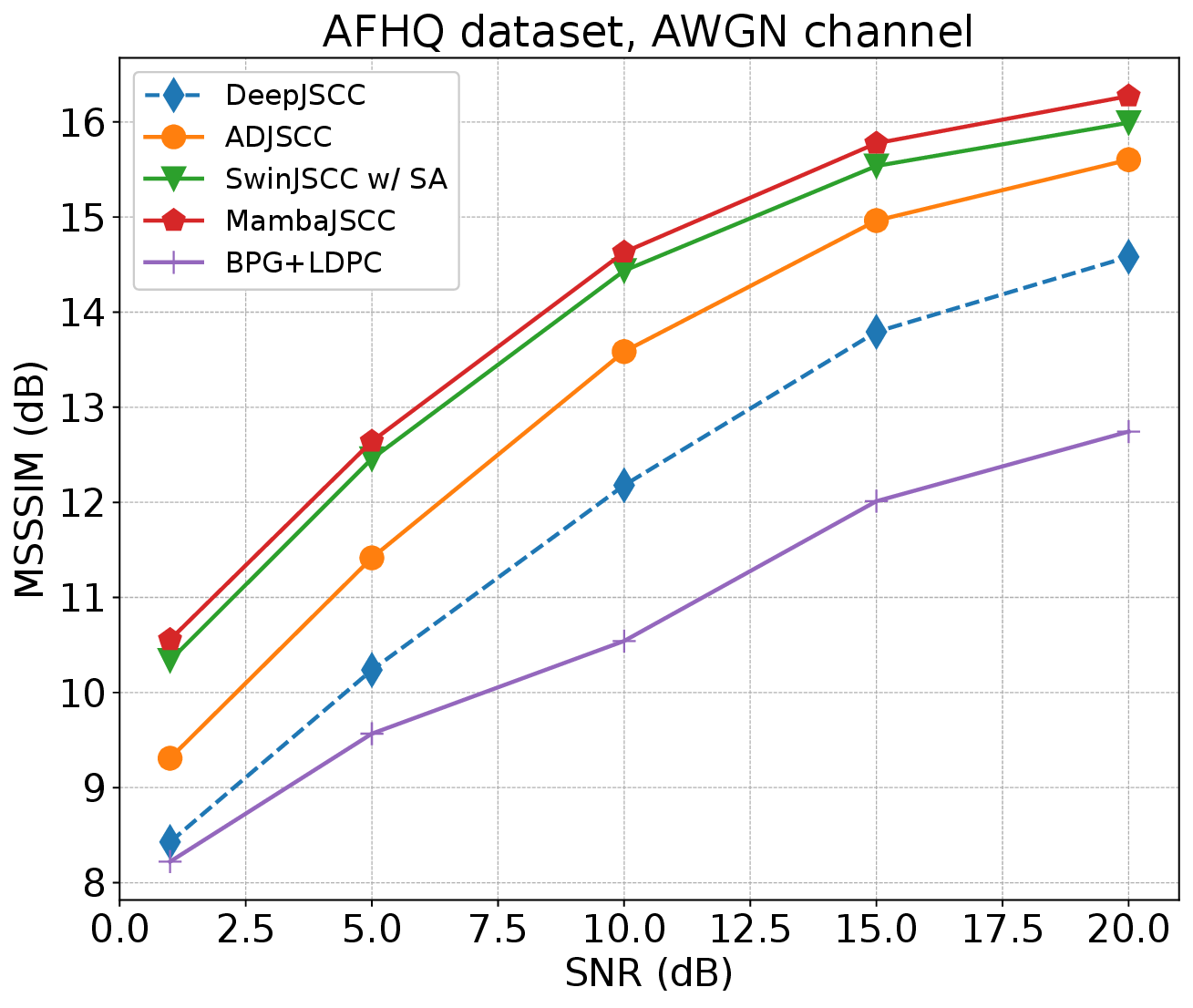}}
    \subfigure[]{\label{AFHQc}\includegraphics[width=0.325\textwidth]{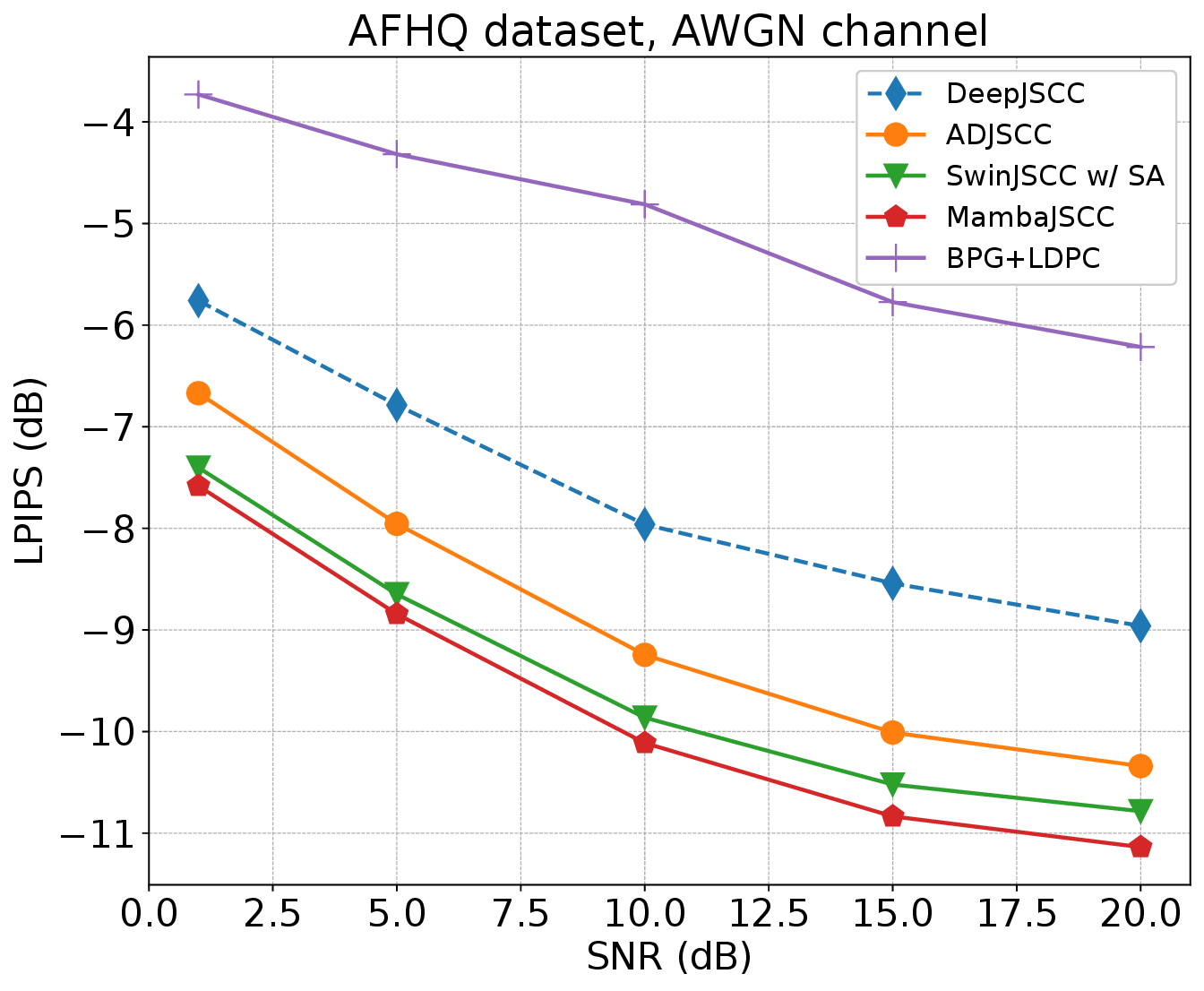}}
    \subfigure[]{\label{AFHQd}\includegraphics[width=0.325\textwidth]{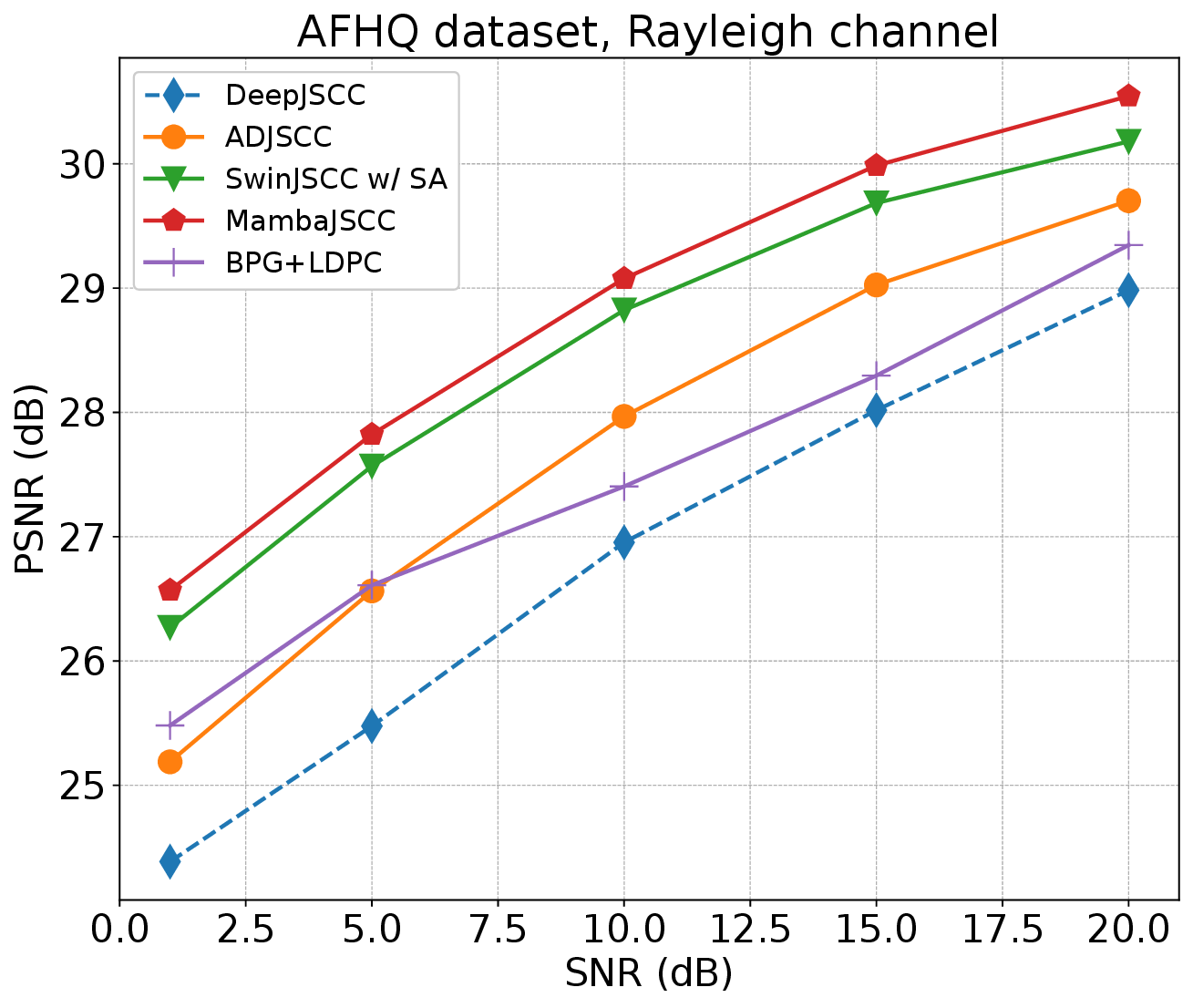}}
    \subfigure[]{\label{AFHQe}\includegraphics[width=0.325\textwidth]{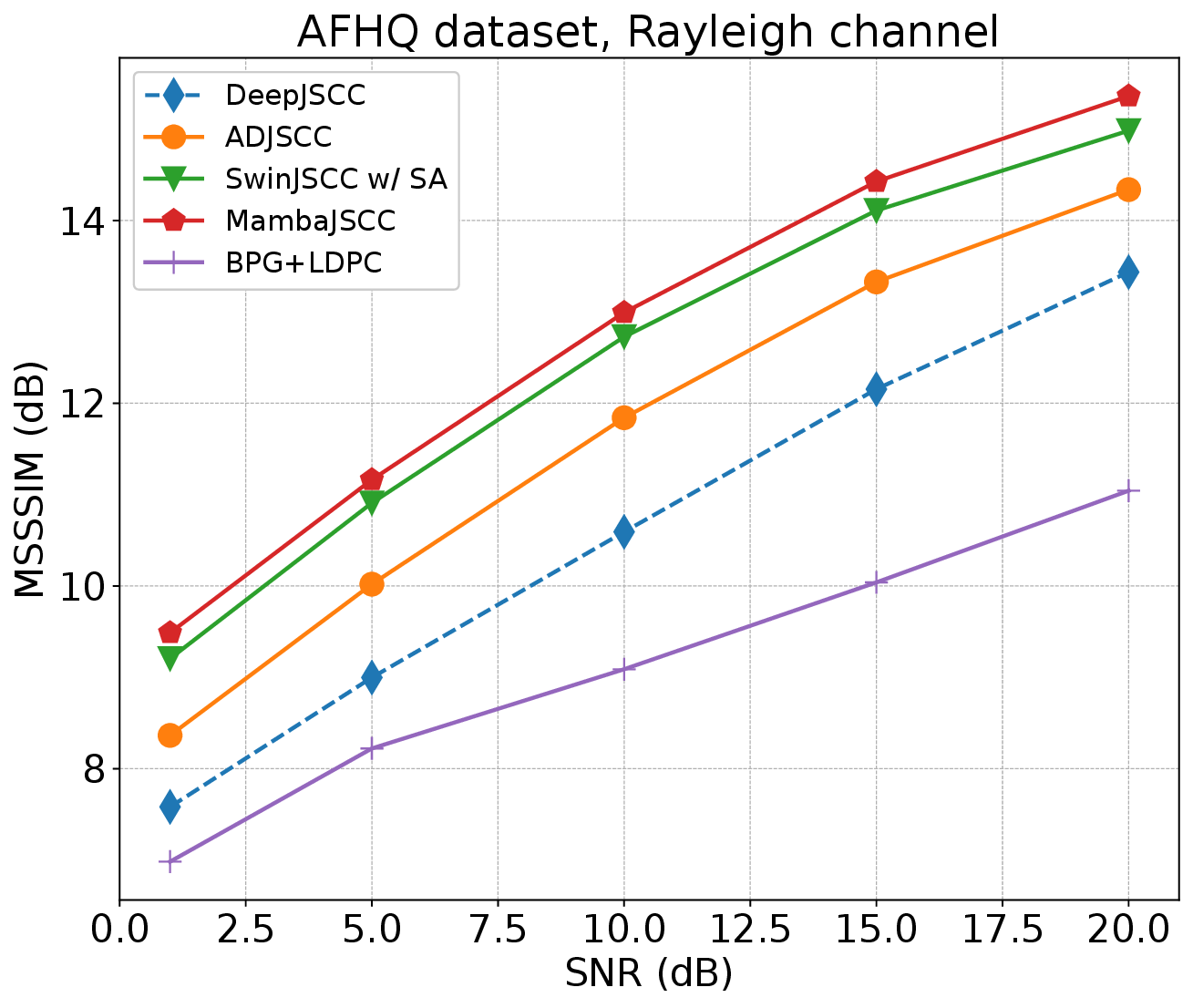}}
    \subfigure[]{\label{AFHQf}\includegraphics[width=0.325\textwidth]{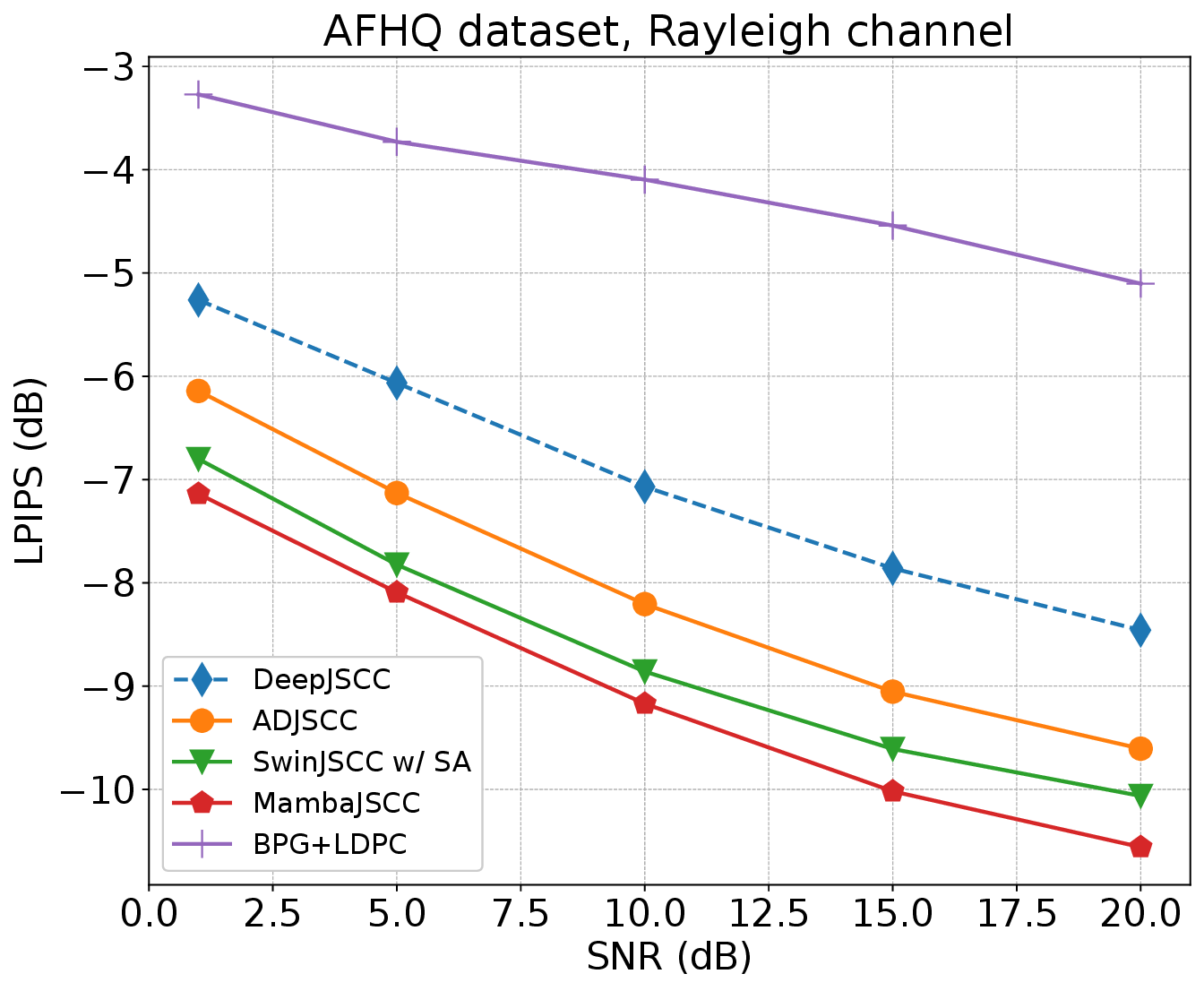}}
    \caption{(a)$\sim$(c) PSNR, MSSSIM and LPIPS performance of different models versus the SNR under the AWGN channel of AFHQ dataset. (e)$\sim$(f) PSNR, MSSSIM and LPIPS performance versus the SNR under the Rayleigh fading channel of AFHQ dataset. The CBR is $\frac{1}{48}$.}
    \label{AFHQ}
    \vspace{-0.3 cm}
    \end{figure*}

To provide a comprehensive evaluation of MambaJSCC, we adopt three datasets with different content and resolutions: CelebA, DIV2K and AFHQ. CelebA contains 150000 images of human faces from well-known figures, DIV2K comprises 800 2K resolution images from a variety of real-world scenes, and AFHQ includes 15000 high-quality images of animal faces for training. For testing, these datasets provide an additional 15000 images, 100 images, and 1500 images, respectively. During  both training and testing, images from CelebA, DIV2K and AFHQ datasets are cropped into varying resolutions of $128\times 128$, $256\times 256$, and $512\times 512$.

Performance is comprehensively evaluated using the PSNR for pixel-wise distortion in reconstruction loss, multi-scale structural similarity index matrices (MSSSIM)\cite{MSSSIM}, and learned perceptual image patch similarity (LPIPS)\cite{LPIPS} for reconstruction loss from the perspective of human perception. For better comparison, the MSSSIM and LPIPS are converted to $dB$ form using the formulas $MSSSIM~(dB)=-10\lg(1-MSSSIM)$ and $LPIPS~(dB)=10\lg(LPIPS)$. In addition to performance metrics, computational complexity is assessed by MACs, which count the number of multiplications and additions required for inference. We also supplement the evaluation by measuring inference delay. Model size is quantified by the number of parameters. MACs are calculated using the Torch operation counter library, and ID is measured on an NVIDIA 4090 GPU.

We conduct a comparative analysis between MambaJSCC and several significant JSCC schemes. The main comparison schemes are as follows: 1) \textbf{MambaJSCC}: The proposed MambaJSCC, utilizing two GSSM modules with CSI-ReST; 2) \textbf{MambaJSCC w/o CA}: The proposed MambaJSCC with the two GSSM modules, but without CSI-ReST; 3) \textbf{SwinJSCC w/ SA}: SwinJSCC with the Channel ModNet, but without the Rate ModNet; 4) \textbf{SwinJSCC w/o SA\&RA}: SwinJSCC without both the Channel ModNet and Rate ModNet; 5) \textbf{ADJSCC}: ADJSCC based on a CNN model with soft-attention that integrates CSI for channel adaptation; 6) \textbf{DeepJSCC}: DeepJSCC, established using pure CNN modules, the first neural network based JSCC model; 7) \textbf{BPG+LDPC}: The classical separation-based source and channel coding schemes, using BPG\cite{BPG} for source coding and 5G NR LDPC\cite{5GNRLDPC} for channel coding. For the BPG+LDPC scheme, we use 5G NR LDPC with code lengths of 8448 and 3840, employing the most suitable code rate and quadrature amplitude modulations. Due to its discrete compression ratio, code rate, and modulation orders, the BPG+LDPC scheme cannot maintain the same channel use as the JSCC models. Therefore, in the simulation, it occupies similar channel resources. MambaJSCC, SwinJSCC w/ SA and ADJSCC incorporate CSI-related designs for channel adaptation, allowing them to be trained across a wide range of SNRs and evaluated with fixed parameters at varying SNRs. In contrast, MambaJSCC w/o CA, SwinJSCC w/o SA\& RA and DeepJSCC are trained and evaluated at fixed SNRs, requiring different parameters for different SNR conditions.

For the basic configurations of MambaJSCC, we use four stages in both the encoder and decoder, with $[n_1^e,n_2^e,n_3^e,n_4^e]=[n_1^d,n_2^d,n_3^d,n_4^d]=[2,2,6,2]$ and $[c_1^e,c_2^e,c_3^e,c_4^e]=[c_1^d,c_2^d,c_3^d,c_4^d]=[128,192,256,320]$. To adjust the model scales, we fix the block number in Stage 3  at $n_3^e=n_3^d$ and vary their values $N_m$. All models are trained using the Adam optimizer until convergence on NVIDIA 4090 GPUs, with an initial learning rate of $10^{-4}$ and the same batch size across models. The loss function used depends on the evaluation metric: mean square error (MSE) loss for PSNR performance, 1-MSSSIM for MSSSIM performance, and LPIPS loss for LPIPS performance.

\subsection{Results Analysis}
In this subsection, we compare the proposed MambaJSCC with other schemes to verify the advantages in terms of performance, computational complexity, and parameter size. Fig. \ref{visual} visualizes the reconstruction results of different schemes over the Rayleigh fading channel at SNR=$5$ dB. As shown, the visual reconstruction quality of MambaJSCC is the best among the five schemes. For example, in the CelebA dataset (first row), the image produced by MambaJSCC closely matches the original image, particularly in the color of the right side of the face, whereas the image from SwinJSCC w/ SA has a reddish tint. The ADJSCC image shows some fake details in the mouth and eyebrow regions. The images from DeepJSCC and the BPG+LDPC scheme are extremely blurred. Similarly, in the DIV2K dataset, MambaJSCC yields the highest quality reconstruction, while the other images progressively degrade in clarity. In the AFHQ dataset, the image from MambaJSCC remains the most details around the cat's face compared to the other schemes.

Fig. \ref{CelebA}-\ref{AFHQ} present the PSNR, MSSSIM, and LPIPS performance of the seven main schemes on the CelebA, DIV2K and AFHQ datasets versus SNR, under both AWGN and Rayleigh fading channels. Solid lines represent the performance of models with channel adaptation, while dashed lines represent models trained for specific SNR values. We can see that the DeepJSCC and ADJSCC schemes perform the worst among the six JSCC-based schemes across all three datasets and evaluation matrices. For example, compared to MambaJSCC at SNR=$1$ dB under the Rayleigh fading channel, the DeepJSCC and ADJSCC schemes perform $2.43$ and $1.35$ dB worse in PSNR on CelebA dataset, $1.94$ and $1.14$ dB worse in MSSSIM on DIV2K dataset, and $1.88$ and $1.00$ dB in LPIPS on AFHQ dataset. The classical separated-based method performs significantly worse in most scenarios, particularly under Rayleigh fading channels, at low SNRs, with low-resolution images, and in terms of the MSSSIM and LPIPS metrics. It only performs comparably to MambaJSCC on high-resolution images of the AFHQ dataset under the AWGN channel at SNR=$20$ dB in terms of PSNR. However, despite these constraints, the BPG+LDPC scheme still underperforms compared to the specialized MambaJSCC w/o CA, further demonstrating the advantages of MambaJSCC over seperated-based methods. 

When compared to the advanced SwinJSCC w/ SA, MambaJSCC outperforms it in all scenarios with the same structure. For example, in terms of PSNR, MSSSIM, and LPIPS, MambaJSCC achieves up to $0.51$, $0.61$, $0.34$ dB gains under the AWGN channel, and $0.52$, $0.57$, $0.34$ gains under the Rayleigh fading channel on CelebA dataset. On the DIV2K dataset, the corresponding performance gains are $0.52$, $0.58$, $0.41$ dB under the AWGN channel, and $0.63$, $0.43$, $0.42$ under the Rayleigh fading channel. For the AFHQ dataset, these gains are $0.56$, $0.29$, $0.35$ dB under the AWGN channel, and $0.36$, $0.38$, $0.50$ dB under the Rayleigh fading channel. Furthermore, the gains are relatively significant. For example, the performance improvements in PSNR, MSSSIM, and LPIPS between the Transformer-based SwinJSCC compared and the CNN-based ADJSCC are only $0.63$, $0.49$, and $0.29$ dB, respectively, on the DIV2K dataset under the Rayleigh fading channel at SNR=$20$ dB. However, MambaJSCC achieve similar or even larger gains compared to SwinJSCC, further highlighting its superior performance. 

On the other hand, MambaJSCC w/o CA and SwinJSCC w/o SA\&RA, which are trained and evaluated separately for each SNR, show better performance than MambaJSCC and SwinJSCC w/ SA because their parameters are specially trained for each SNR, whereas MambaJSCC and SwinJSCC are trained for a range of SNRs. Nonetheless, MambaJSCC experiences only minimal performance loss with a single model, particularly at low SNRs. For example, the PSNR, MSSSIM, and LPIPS performance gaps between MambaJSCC and MambaJSCC w/ CA under the Rayleigh fading channel at SNR=$10$ dB are only $0.11$, $0.10$, and $0.10$ dB on CelebA datasets. Similar performance gaps are observed in other scenarios, demonstrating that MambaJSCC with the proposed CSI-ReST channel adaptation method, as a universal model, can achieve performance comparable to  the specialized model. Moreover, Fig. \ref{CBRlawPSNR} shows the PSNR performance on the DIV2K dataset versus CBR, confirming that the performance gains are stable across varying CBR values. The consistent performance gains across all three evaluation metrics, under both channels, and across different  SNRs and CBRs on the three datasets, strongly demonstrate that MambaJSCC consistently achieves the best performance among the existing main JSCC schemes. 

In addition to delivering excellent performance, MambaJSCC also offers significant lightweight advantages, as shown in Table \ref{tab1}. It requires only 72\% of the MACs and 51\% of the parameters needed by SwinJSCC w/ SA for inference. As a supplement to the MACs, we also provide the ID on the DIV2K dataset. Although the DeepJSCC and ADJSCC schemes show some room for performance improvement, their parameter size and the MACs of DeepJSCC are relatively small. Overall, the experimental results clearly demonstrate the advantages of the Mamba model, which not only achieves outstanding performance but also maintains low computational and parameters overhead.

It is well known that the deep learning models follow the scaling law, which suggests that performance may improve as the parameter size increases. Therefore, we modify the structure by adding more VSSM-CA blocks to create a large model and compare its performance with SwinJSCC w/ SA, illustrating the potential performance gains at a similar parameter cost, in accordance with the scaling law. Specifically, we increase the $N_m$ from $6$ to $17$, which raises the parameters count from $14.54$M to $28.09$M, comparable to the $28.24$M parameters of SwinJSCC w/ SA. The performance under both channels on the DIV2K dataset is shown in Fig. \ref{scaling}. For example, the LPIPS performance gain increases from $0.41$ dB to $0.90$ dB under the AWGN channel at an SNR of $20$ dB, further demonstrating the significant performance boost of MambaJSCC.

\subsection{Ablation Experiment}

  % \begin{figure}[h]
  %   \centering
  %   \includegraphics[width=0.5\textwidth]{./figs/DIV2K/channel_adaptive_PSNR.eps}
  %   \caption{PSNR performance of models with different channel adaptation methods versus the SNR under the AWGN and Rayleigh fading channels on DIV2K dataset. The CBR is set to $\frac{1}{48}$.}
  %   \label{CA_DIV2K}
  %   \end{figure}

% \begin{figure}[ht]
%   \centering
%   \includegraphics[width=0.5\textwidth]{./figs/AFHQ/channel_adaptive_PSNR.eps}
%   \caption{PSNR performance of models with different channel adaptation methods versus the SNR under the AWGN and Rayleigh fading channels on AFHQ dataset. The CBR is set to $\frac{1}{48}$.}
%   \label{CA_AFHQ}
%   \end{figure}
In this subsection, we remove specific modules from MambaJSCC to evaluate their individual contributions to performance. First, we demonstrate the effectiveness of CSI-ReST for channel adaptation. We replace CSI-ReST in MambaJSCC with the Channel ModNet proposed in SwinJSCC, referring to this variant as MambaJSCC w/o CA with Channel ModNet. We also use MambaJSCC w/o CA but modify its training scheme to cover a wide range of SNRs, ensuring consistency with the training scheme of other models. 

\begin{figure}[t]
  \centering
  \includegraphics[width=0.42\textwidth]{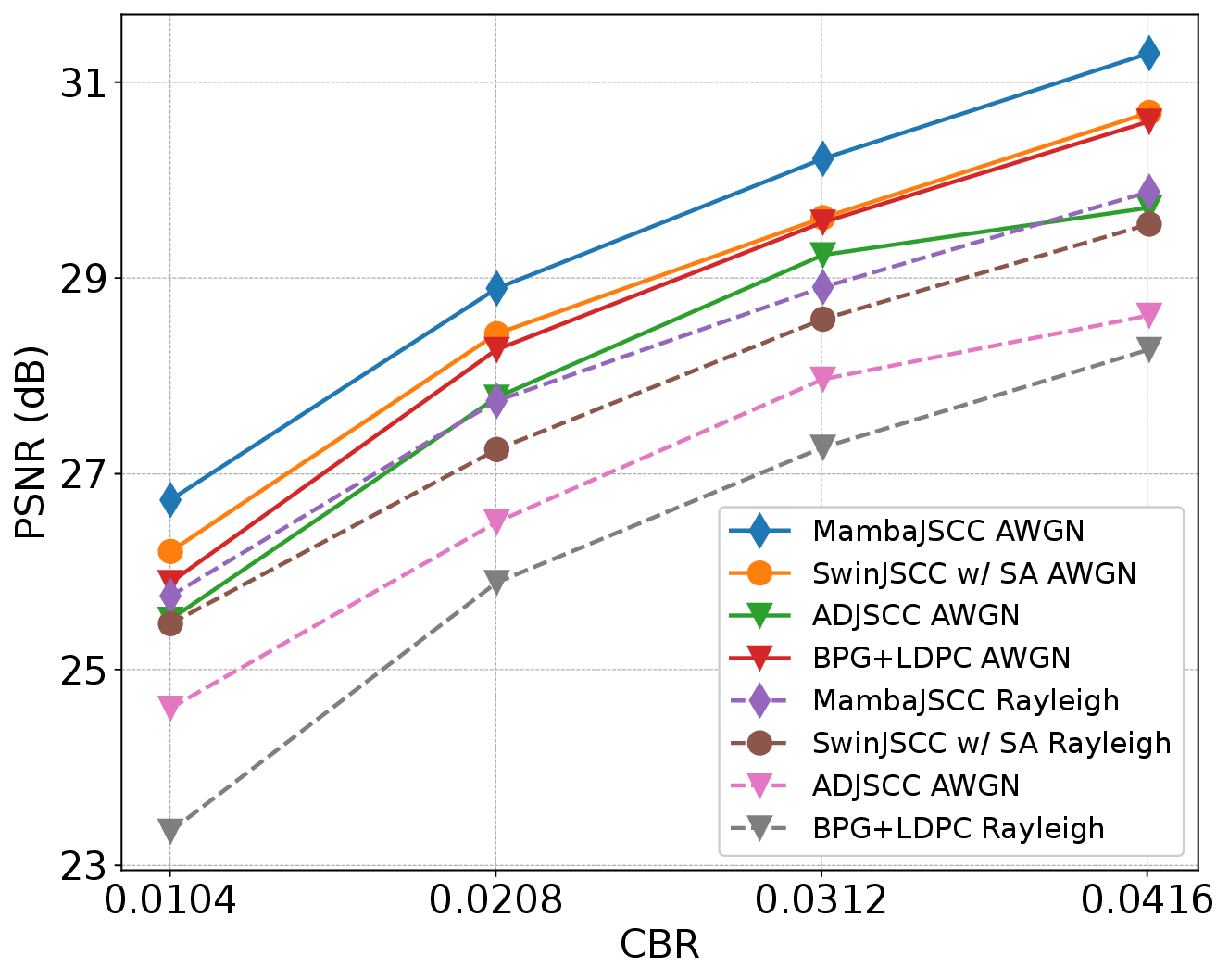}
  \caption{The PSNR of different models with channel adaptation ability versus the CBRs under both the AWGN and Rayleigh fading channels on the DIV2K dataset. The SNR is set to $15$ dB.}
  \label{CBRlawPSNR}
  \vspace{-0.5 cm}
  \end{figure}

\begin{figure}[t]
  \centering
  % \subfigure[]{\label{DIV2KscalePSNR}\includegraphics[width=0.325\textwidth]{./figs/DIV2K/scale_PSNR.eps}}
  % \subfigure[]{\label{DIV2KscaleMSSSIM}\includegraphics[width=0.325\textwidth]{./figs/DIV2K/scale_MSSSIM.eps}}
  {\includegraphics[width=0.42\textwidth]{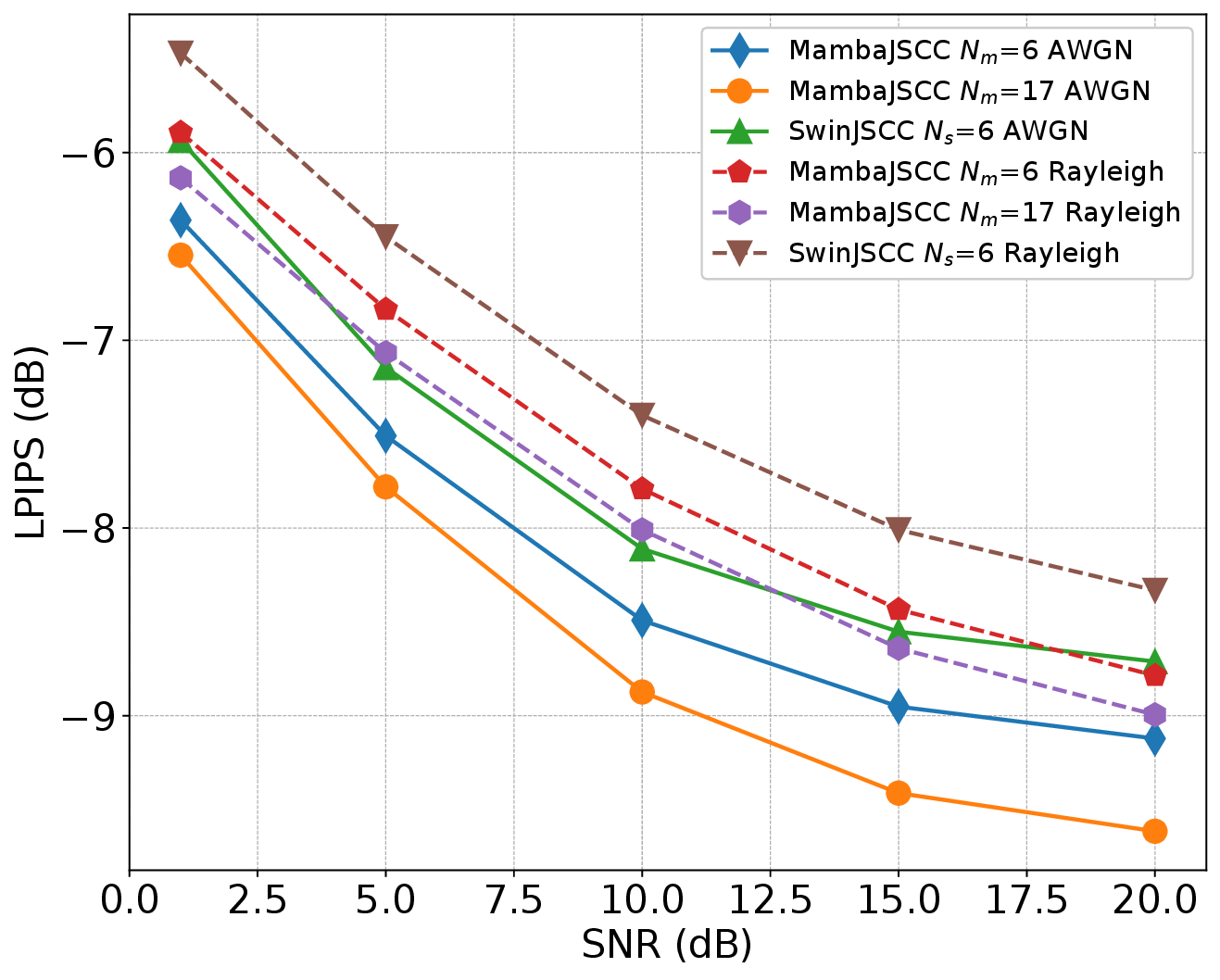}}
  \caption{LPIPS performance of models with different block numbers versus the SNR under the AWGN and Rayleigh fading channels on the DIV2K dataset. The CBR is set to $\frac{1}{48}$.}% The MambaJSCC are equipped with more VSSM-CA blocks to match the parameter size of SwinJSCC w/ SA.}
  \label{scaling}
  \vspace{-0.3 cm}
  \end{figure}
\begin{figure*}[t]
  \centering
  \subfigure[]{\includegraphics[width=0.325\textwidth]{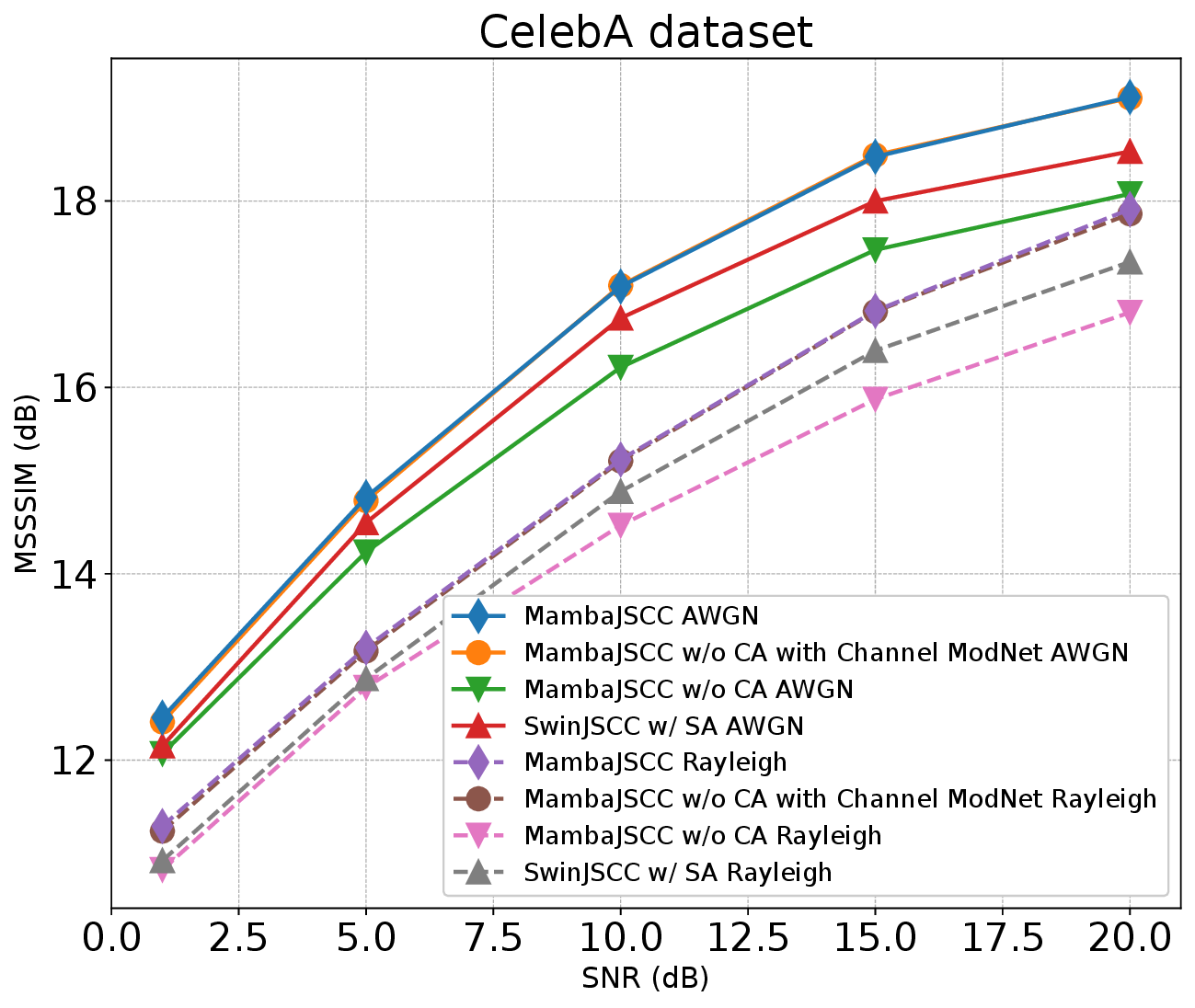}}
  \subfigure[]{\includegraphics[width=0.335\textwidth]{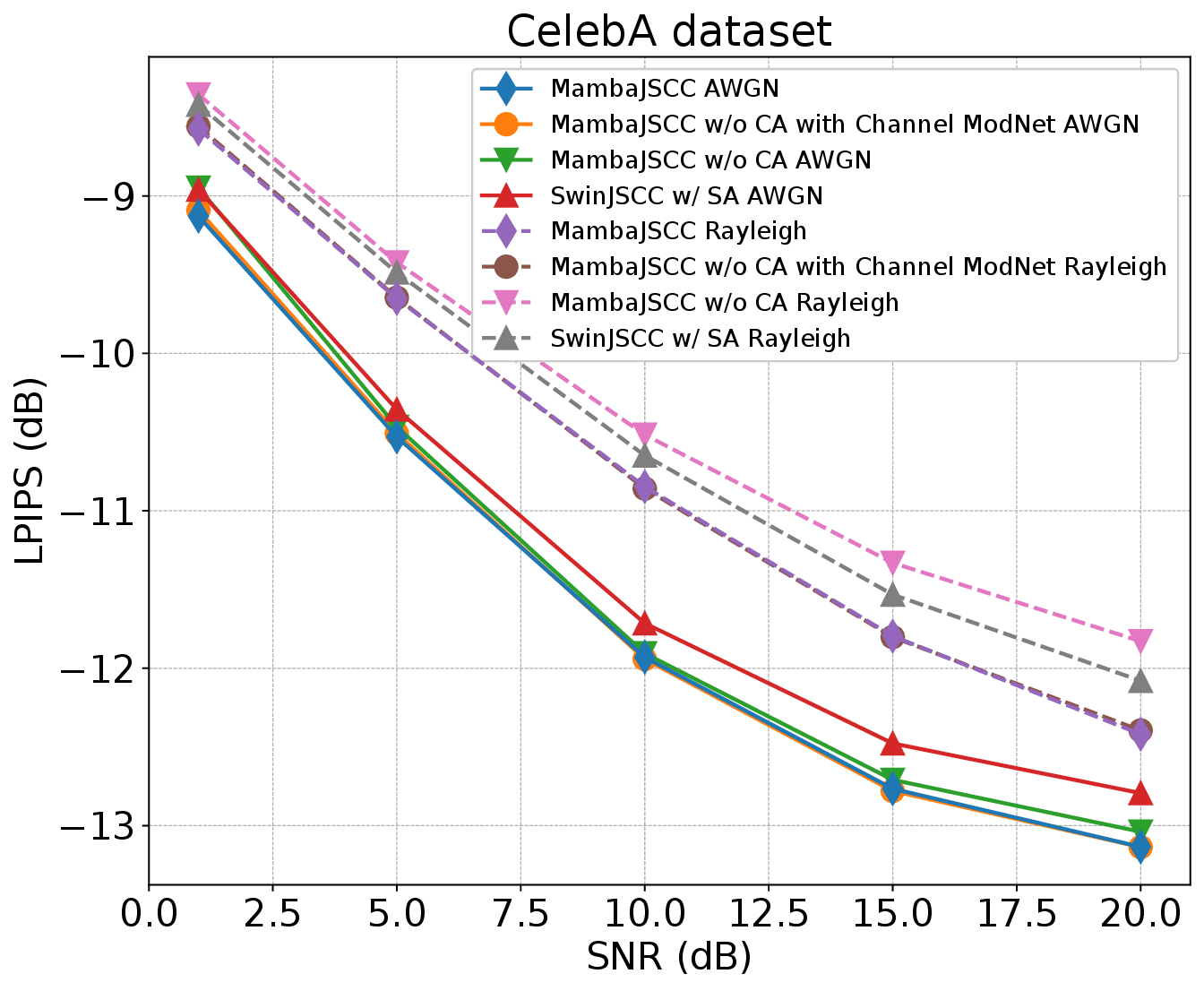}}
  \subfigure[]{\includegraphics[width=0.325\textwidth]{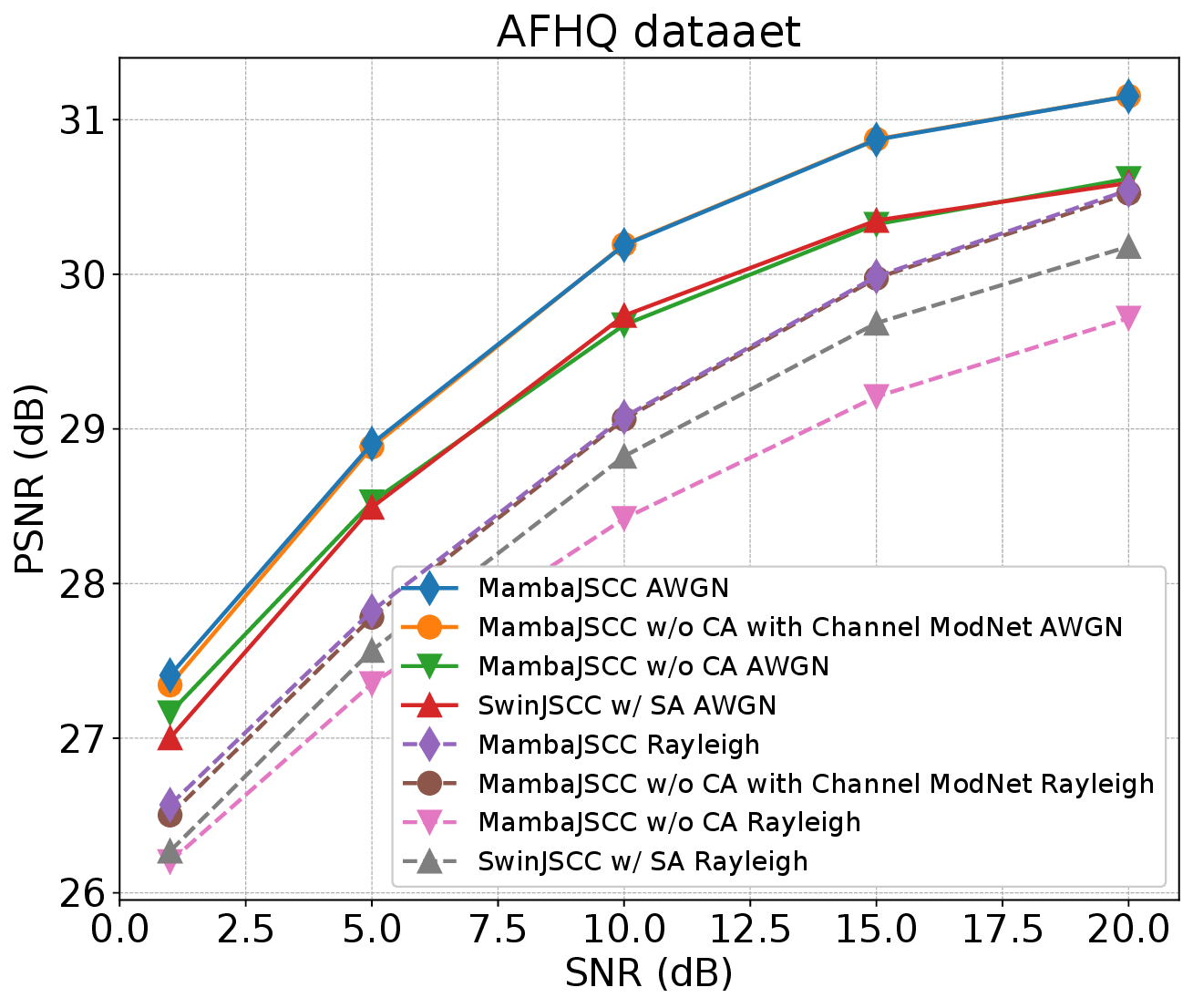}}
  \caption{MSSSIM, LPIPS and PSNR performance of models with different channel adaptation methods versus the SNR under the AWGN and Rayleigh fading channels on the CelebA and AFHQ datasets. The CBR is set to $\frac{1}{48}$.}
  \label{CelebA_adp}
  \end{figure*}
  \begin{table}[t]
    \renewcommand\arraystretch{1.5}
    \centering
    \caption{{MACs, ID and Parameters of different models on the DIV2K and AFHQ datasets. The CBR is set to $\frac{1}{48}$}.}
    \label{tab1}
    \resizebox{\linewidth}{!}{
    \begin{tabular}{|c|c|c|c|c|c|c|}\hline
    %|c|c|c|c|c||c|c|c|c|c|
    \multirow[vpos]{2}{*}{Schemes} & \multicolumn{2}{|c|}{DIV2K} & \multicolumn{2}{|c|}{AFHQ} & \multirow[vpos]{2}{*}{Paras(M)}\\ \cline{2-5}
    & MACs(G) & ID(ms) &MACs(G) & ID(ms) & \\
    \hline
    {MambaJSCC}      & 25.76 & 6.11 & 103.04 &25.10 &14.54\\\hline
    {SwinJSCC w/ SA} & 35.54 & 6.65 & 142.14 &26.40 &28.24\\\hline
    {ADJSCC}         & 46.21 & 7.67 & 184.85 &11.87 &10.57\\\hline
    {DeepJSCC}       & 0.72  & 0.71 & 2.88   &1.16  &0.16 \\\hline
     % #364
    \end{tabular}}
    \vspace{-0.3 cm}
    \end{table}
\begin{table}[t]
  \renewcommand\arraystretch{1.5}
  \centering
  \caption{{MACs, Parameters and the PSNR performance on the AWGN channel with SNR=$20$ dB of MambaJSCC with different module on the CelebA dataset. The CBR is set to $\frac{1}{48}$.}}
  \label{tab2}
  \resizebox{\linewidth}{!}{
  \begin{tabular}{|c|c|c|c|}\hline
  %|c|c|c|c|c||c|c|c|c|c|

  \multicolumn{1}{|c|}{Models}& MACs (G) & Parameters (M) & PSNR (dB)\\
  \hline
  MambaJSCC & 6.44 & 14.54 & 31.29\\\hline
  MambaJSCC w/o CA & 6.44 & 14.54 & 30.39 \\\hline
  {MambaJSCC w/o CA} & \multirow[vpos]{2}{*}{6.44+0.22}  &\multirow[vpos]{2}{*}{14.54 + 9.87}& \multirow[vpos]{2}{*}{31.28}\\ 
  {with Channel ModNet} & & &\\\hline

   % #364
  \end{tabular}}
  \vspace{-0.3 cm}
  \end{table}

  \begin{table}[t]
    \renewcommand\arraystretch{1.5}
    \centering
    \caption{{The PSNR, MACs and Parameters of MambaJSCC with different modules on the DIV2K dataset. The SNR is $20$ dB.}}
    \label{tab3}
    \resizebox{\linewidth}{!}{
    \begin{tabular}{|c|c|c|c|c|c|}\hline
    
    {$N_s$}& MLP & GSSM number & PSNR(dB) & MACs(G) &Para(M)\\
    \hline
    \multirow[vpos]{3}{*}{$6$} & \checkmark & 2 & 29.26 & 25.76 & 14.54 \\ 
     & \usym{2717} &2 &  29.17 &20.54 & 11.70 \\
     & \checkmark &4 & 29.26 & 29.03& 16.38 \\\hline
    17 & \checkmark & 2 & 29.49 & 39.21 & 28.09 \\ \hline 
    \end{tabular}}
    \vspace{-0.5 cm}
    \end{table}
As shown in Fig. \ref{CelebA_adp}, MambaJSCC significantly outperforms MambaJSCC w/o CA on the CelebA dataset, with maximum gains of $1.03$ in MSSSIM and $0.35$ dB in LPIPS under the AWGN channel, and $1.10$ and $0.60$ dB, respectively, under the Rayleigh fading channel. In contrast, MambaJSCC w/o CA with Channel ModNet achieves nearly the same performance as MambaJSCC.

The same phenomenon is observed on the AFHQ dataset, as shown in Fig. \ref{CelebA_adp}(c). These experimental results demonstrate that the CSI-ReST channel adaptation method significantly improves the performance of a single model across varying SNRs and achieves the same channel adaptation capability as the Channel ModNet method. Moreover, despite delivering comparable performance, the CSI-ReST method is a zero-parameter, zero-computational method designed by utilizing the feature of GSSM, whereas the Channel ModNet is a plug-in module that requires additional computational and parameter overhead. As shown in Table \ref{tab2}, MambaJSCC maintains the same MACs and parameters count as MambaJSCC w/o CA, yet improves performance by $0.90$ dB. In contrast, Channel ModNet introduces an additional $0.22$G MACs on the CelebA dataset with 9.87M parameters. 

Secondly, we increase the number of GSSM modules addressing each sequence from two to four by adding two additional scan directions: one expanding the patch in the hierarchical direction and the other in its reverse direction, as proposed in \cite{VSSM}. As shown in Table \ref{tab3}, despite using four GSSM modules, the performance remains at $29.26$ while the MACs and parameter count increase by $3.27$G and $1.84$M, respectively. This fact confirms Proposition \ref{pro4}, which two GSSM modules are sufficient to capture global information. Thirdly, we remove the channel MLP from the VSSM-CA module, resulting in a $0.1$ dB performance drop, but also reduces the computation overhead by $5.22$G MACs and the parameter count by $2.84$M. In comparison to the scaling law experiment, increasing $N_s$ to $17$ only achieved a $0.23$ dB gain, with an additional $13.45$G MACs and $13.55$M parameters. A six-fold increase in parameters only yielded a 2.3-fold performance improvement. Thus, the inclusion of the channel MLP module is an effective design choice.

%All the experiments prove our MambaJSCC is advanced in performance, computational and parameters overhead and the propositions and design in our paper are reasonable.
%, proving that the CSI-ReST method succeeds in improving the channel adaptation ability
%We can see that MambaJSCC w/o CA requires $15.30$ billion fewer MACs and $8.47$ million fewer parameters compared with SwinJSCC w/o SA\&RA. Meanwhile, we can see from Table I that the ID of MambaJSCC w/o CA is 39.12 ms shorter than that of SwinJSCC w/o SA\&RA. Even with significant reductions in both MACs and parameters, and a lower ID, the proposed MambaJSCC w/o CA still outperforms SwinJSCC w/o SA\&RA in terms of PSNR. For example, on the Rayleigh fading channel, Mamba JSCC w/o CA achieves a $0.53$ dB gain at an SNR of $5$ dB. For the AWGN channel, Mamba JSCC w/o CA achieves a $0.97$ dB gain at an SNR of $1$ dB.

%For conclusion, the our MambaJSCC models achieve significant performance gains and the proposed CSI embedding method achieves comparable performance compared with their benchmarks. Meanwhile our models maintain  extremely low model size and computational complexity, enabling our MambaJSCC a promising JSCC architecture in real-time, energy and storage limited applications in communication systems.

\section{CONCLUSION}\label{V}
In this paper, we propose MambaJSCC, a novel lightweight and efficient JSCC architecture for wireless image transmission. By designing GSSM modules expressed by reversible matrix transformations, we prove that MambaJSCC effectively captures global information while maintaining low complexity. Moreover, we develop the zero-parameter, zero-computation CSI-ReST method to harness the endogenous capability of GSSM for channel adaptation. Extensive experiments confirm that MambaJSCC outperforms existing JSCC architectures, including SwinJSCC, in terms of distortion and perception performance, parameter size, computational overhead and inference delay. For example, when using an equal number of VSSM-CA blocks and Swin Transformer blocks as SwinJSCC, MambaJSCC achieves a $0.52$ dB gain in PSNR and $0.41$ dB gain in LPIPS while requiring only 72\% MACs, 51\% of the parameters, and 91\% of ID.

\footnotesize
\bibliographystyle{IEEEtran}
\bibliography{reference}{}
\end{document}